\documentclass[11pt]{article} 

\usepackage{geometry}
\usepackage{hyperref}
\usepackage{microtype}

\usepackage{graphicx}
\usepackage{amssymb}
\usepackage{amsmath}
\usepackage{amsthm} 
\usepackage{amsfonts}
\usepackage[ruled]{algorithm2e}
\usepackage{tikz}
\usetikzlibrary{arrows}
\usepackage{color}
\definecolor{darkblue}{rgb}{0,0,0.4}

\newcommand {\mm}[1] {\ifmmode{#1}\else{\mbox{\(#1\)}}\fi}
\newcommand{\dime}[1]       {\mm{\rm dim\,}{#1}}
\newcommand{\Lcal}{\mathcal{L}}
\newcommand{\tLcal}{\tilde{\mathcal{L}}}
\newcommand{\tW}{\tilde{W}}
\newcommand{\Rspace}        {\mm{{\mathbb R}}}
\newcommand{\Sspace}        {\mm{{\mathbb S}}}

\newcommand{\Hom}        {\mm{\mathrm{Hom}}}

\newcommand{\pinv}{{+}}
\newcommand{\Exp}{\mathbb{E}}
\newcommand{\nnz}{\mm{nnz}}

\newtheorem{theorem}{Theorem}[section]

\newtheorem{proposition}{Proposition}[section]
\newtheorem{corollary}{Corollary}[section]

\newcommand{\denselist}{\vspace{-5pt} \itemsep -2pt\parsep=-1pt\partopsep -2pt}
\newcommand{\para}[1]        {\noindent{\textbf{#1}}}

\DeclareMathOperator{\sgn}{sgn}

\title{Spectral Sparsification of Simplicial Complexes \\
    for Clustering and Label Propagation}

\author{Braxton Osting\thanks{E-mail:~osting@math.utah.edu.} \\ University of Utah 
\and Sourabh Palande\thanks{E-mail: sourabh@sci.utah.edu.}\\University of Utah
\and Bei Wang\thanks{E-mail: beiwang@sci.utah.edu.} \\ University of Utah}
\date{}

\begin{document}

\maketitle

	
\begin{abstract}
    As a generalization of the use of graphs to describe pairwise interactions, simplicial complexes can be used to model higher-order interactions between three or more objects in complex systems.
    There has been a recent surge in activity for the development of data analysis methods applicable to simplicial complexes, including techniques  based on computational topology, higher-order random processes, generalized Cheeger inequalities, isoperimetric inequalities and spectral methods.
    In particular, spectral learning methods (e.g. label propagation and clustering) that directly operate on simplicial complexes represent a new direction for analyzing such complex datasets.
    
    To apply spectral learning methods to massive datasets modeled as simplicial complexes, we develop a method for sparsifying simplicial complexes that preserves the spectrum of the associated Laplacian matrices. 
    We show that the theory of Spielman and Srivastava for the sparsification of graphs extends to simplicial complexes via the up Laplacian. 
    In particular, we introduce a generalized effective resistance for simplices, provide an algorithm for sparsifying simplicial complexes at a fixed dimension, and give a specific version of the generalized Cheeger inequality for weighted simplicial complexes. 
    Finally, we introduce higher-order generalizations of spectral clustering and label propagation for simplicial complexes and demonstrate via experiments the utility of the proposed spectral sparsification method for these applications.
\end{abstract}


\section{Introduction}

    Understanding massive systems with complex interactions and multi-scale dynamics is  important  in a variety of social, biological and technological settings.
    A commonly-used approach to understanding such a system is to represent it as a graph where vertices represent objects and (weighted) edges represent \emph{pairwise interactions} between the objects.
    A large arsenal of methods has been developed to analyze properties of graphs, which can then be combined  with domain-specific knowledge to infer properties of the system being studied.
    These tools include graph partitioning and clustering~\cite{OstingWhiteOudet2014,GennipGuillenOsting2014,Luxburg2007}, random processes on graphs~\cite{Gleich2015}, graph distances, various measures of graph connectivity~\cite{OstingBruneOsher2014}, combinatorial graph invariants~\cite{Diestel2000} and spectral graph theory~\cite{Chung1997}.
    In particular, spectral methods for graph-based learning have had great success due to their efficiency and good theoretical guarantees for applications ranging from image segmentation~\cite{MahoneyOrecchiaVishnoi2012} to community detection~\cite{AndersenLang2006}.
    For example, the spectral clustering method (see, e.g.,~\cite{AndersenChungLang2006, SpielmanTeng2013}) is a graph-based learning method  used for the unsupervised clustering task and label propagation~\cite{SzummerJaakkola2002,ZhuGhahramaniLafferty2003} is a graph-based learning method for semi-supervised regression.

    \para{Simplicial complexes and data analysis.} 
        While graphs have been used with great success to describe pairwise interactions between objects in datasets, they fail to capture \emph{higher-order interactions} that occur between three or more objects.
        Higher-order interactions in complex datasets can be modeled using \emph{simplicial complexes} \cite{GradyPolimeni2010,Munkres1984}.
        There has been a recent surge in activity to develop machine learning methods for data represented by simplicial complexes, including methods based on computational topology~\cite{Carlsson2009,EdelsbrunnerHarer2008,Ghrist2008,GradyPolimeni2010}, higher-order random processes~\cite{BensonGleichLeskovec2015,GleichLimYu2015}, generalized Cheeger inequalities~\cite{GundertSzedlak2015,SteenbergenKlivansMukherjee2014}, isoperimetric inequalities~\cite{ParzanchevskiRosenthalTessler2016}, high-dimensional expanders~\cite{DotterrerKahle2012, Lubotzky2014, ParzanchevskiRosenthal2016} and spectral methods~\cite{HorakJost2013}.
        In particular, topological data analysis methods using simplicial complexes as the underlying combinatorial structures have been successfully employed for diverse applications~\cite{SilvaGhrist2007,JiangLimYao2011,LeePedersenMumford2003,NicolauLevineCarlsson2011,Yuan2014,PereaHarer2015,RenZhaoRamanathan2013,WangSummaPascucci2011}.
    
        Learning (indirectly or directly) based on simplicial complexes represents a new direction recently emerging from the confluence of computational topology and machine learning.
        This is ongoing work; while topological features derived from simplicial complexes, used as input to machine learning algorithms, have been shown to increase the predictive power compared to graph-theoretic features~\cite{BendichGasparovicHarer2015, WongPalandeWang2016}, there is still interest in developing learning algorithms that directly operate on simplicial complexes. 
        For example, researchers have begun to develop mathematical intuition behind higher-dimensional notions of spectral clustering and label propagation~\cite{MukherjeeSteenbergen2016,SteenbergenKlivansMukherjee2014,Luxburg2007}.
    
    \para{Sparsification of graphs and simplicial complexes.}
        For unstructured graphs representing massive datasets, the computational costs associated with na\"ive implementations of many graph-based algorithms is prohibitive.
        In this scenario, it is useful to approximate the original graph with one having fewer edges or vertices while preserving certain properties of interest, known as \emph{graph sparsification}.
        A variety of graph sparsification methods have been developed that allow for both efficient storage and computation \cite{BenczurKarger1996,SpielmanSrivastava2011,SpielmanTeng2011}.
        In particular, in seminal work, Spielman and Srivastava developed a method for sparsifying graphs that approximately preserves the spectrum of the graph Laplacian~\cite{SpielmanSrivastava2011}.
        It is well-known from spectral graph theory that the spectrum of the graph Laplacian bounds a variety of properties of interest including the size of cuts (i.e.~bottlenecks), clusters (i.e.~communities), distances, various random processes (i.e.~PageRank) and combinatorial properties (e.g.~coloring, spanning trees, etc.).
        It follows that this method~\cite{BatsonSpielmanSrivastava2013} can be used to produces a sparsified graph that contains a great deal of information about the original graph and hence, in the graph-based machine learning setting, about the underlying dataset.
        
        Analogously, computational methods that operate on simplicial complexes are severely limited by the computational costs associated with massive datasets. While geometric complexes (embedded in Euclidean space) tend to be naturally sparse, abstract simplicial complexes coming from data analysis can be dense and do not have natural embeddings in Euclidean space.
        For example, a dense simplicial complex is obtained when representing funnctional brain netwroks using simplicial complexes (e.g.,~\cite{CassidyRaeSolo2015,LeeChungKang2011,LeeKangChung2012}). 
        Here, a brain network is mapped to a point cloud in a metric space, where network nodes map to points, pairwise associations between nodes map to distances between pairs of points, and higher order information is mapped to  higher dimensional simplices~\cite{AndersonAndersonPalande2018,WongPalandeWang2016}.
        Another motivation behind studying sparsification of simplicial complexes is the fact that high-order tensors (multidimensional arrays) can be represented by  simplicial complexes and vice versa. Just as spectral graph sparsifiers are useful in matrix decompositions and linear system solvers, one can expect simplicial complex sparsifiers to be useful in tensor decompositions and multi-linear system solvers.
        
        Several approaches have been recently proposed to sparsify simplicial complexes.
        One class of methods, referred to as \emph{homological sparsification}, involves constructing a sparse simplicial complex that approximates persistence homology~\cite{BotnanSpreemann2015, BuchetChazalOudot2015, CavannaJahanseirSheehy2015, ChoudharyKerberRaghvendra2016, DeyFanWang2014, DeyFanWang2015, KerberSharathkumar2013,Sheehy2013, TauszCarlsson2011}.
        Persistence homology~\cite{EdelsbrunnerLetscherZomorodian2002} turns the algebraic concept of homology into a multi-scale notion.
        It typically operates on a sequence of simplicial complexes (referred to as a \emph{filtration}), constructs a series of homology groups and measures their relevant scales in the filtration.
        Common simplicial filtrations arise from \v{C}ech or Vietoris-Rips complexes, and most of the homological sparsification techniques produce sparsified  complexes that give guaranteed approximations to the persistent homology of the unsparsified filtration.
        
        The sparsification processes involve either the removal or subsampling of vertices, or edge contractions from the sparse filtration.
        It is also possible to sparsify simplicial complexes using another class of methods called sketching, particularly, those applied to tensors.
        Tensor decomposition methods have found many applications in machine learning~\cite{KoldaBader2009}, including recent advancements in tensor sparsification~\cite{WangTungSmola2015,HauptLiWoodruff2017,NguyenDrineasTran2015,SongWoodruffZhong2017} using sampling methods from randomized linear algebra.
        
        Since many learning methods based on simplicial complexes rely --- either explicitly or implicitly --- on the spectral theory for higher-order Laplacians, it is desirable to develop methods for sparsifying simplicial complexes that approximately preserves the spectrum of higher-order Laplacians.
     
    \para{Contributions.}
        In this paper, motivated by learning based on simplicial complexes, we develop computational methods for the spectral sparsification of simplicial complexes.
        In particular:
        \begin{itemize}\denselist
        	\item We introduce a \emph{generalized effective resistance} of simplices by extending the notion of \emph{effective resistance} of edges (e.g.~\cite{ChandraRaghavanRuzzo1996, DoyleSnell1984,GhoshBoydSaberi2008}); see Section~\ref{sec:algorithm}.  
        	\item We extend the methods and analysis of Spielman and Srivastava~\cite{SpielmanSrivastava2011} for sparsifying graphs to the context  of simiplicial complexes at a fixed dimension. We prove that the spectrum of the \emph{up Laplacian}  is approximately preserved under sparsification in the sense that the spectrum of the up Laplacian for the sparsified simplicial complex is controlled by the spectrum of the up Laplacian for the original simplicial complex; see Theorem \ref{thm:main}.  
        	\item We generalize the Cheeger constant of Gundert and Szedl\'ak for \emph{unweighted} simplicial complexes \cite{GundertSzedlak2015} to  \emph{weighted} simplicial complexes and verify that the Cheeger inequality involving  the first non-trivial eigenvalue of the \emph{weighted} up Laplacian holds in the sparsfied setting; see Proposition~\ref{prop:newCheeger}.  
        	\item Our theoretical results are supported by substantial numerical experiments. By extending spectral learning algorithms such as spectral clustering and label propagation to simplicial complexes,
        	we demonstrate that preserving the structure of the up Laplacian via sparsification also preserves the results of these algorithms (Section~\ref{sec:experiments}). 
        	These applications exemplify the utility of our spectral sparsification methods. 
        \end{itemize}
        We proceed by reviewing  background results and introducing notation in Section~\ref{sec:background} that gives a brief description of relevant algebraic concepts, effective resistance, and spectral sparsification of graphs.
        The theory and algorithm for sparsifying simplicial complexes are presented in Section~\ref{sec:algorithm}. We state the implications of the algorithm for a generalized Cheeger cut for the simplicial complex in Section~\ref{sec:Cheeger}.
        We showcase experimental results validating our algorithms in Section~\ref{sec:experiments} and conclude with a discussion and some open questions in Section~\ref{sec:discussion}.

	
\section{Background}
\label{sec:background}
    
    \para{Simplicial complexes.}
        A \emph{simplicial complex} $K$ is a finite collection of simplices such that every face of a simplex of $K$ is in $K$ and the intersection of any two simplices of $K$ is a face of each of them~\cite{Munkres1984}.
        The $0$-, $1$- and $2$-simplices correspond to vertices, edges and triangles.
        An \emph{oriented simplex} is a simplex with a chosen ordering of its vertices.
        For the remainder of this paper, let $K$ be an oriented simplicial complex on a vertex set $[n] = \{1,2,\ldots, n\}$.
        Let $S_p(K)$ denote the collection of all oriented $p$-simplices of $K$ and $n_p = |S_p(K)|$.
        The \emph{$p$-skeleton} of $K$ is denoted as $K^{(p)} :=\bigcup_{0 \leq i\leq p} S_i(K)$.
        Let $\dime{K}$ denote the dimension of $K$.
        For a review of simplicial complexes, see \cite{Ghrist2014, GradyPolimeni2010, Munkres1984}.
    
    \para{Laplace operators on simplicial complexes.}
        The $i$-th \emph{chain group} $C_i(K) = C_i(K, \Rspace)$ of a complex $K$ with coefficient $\Rspace$ is a vector space over the field $\Rspace$ with basis $S_i(K)$.
        The $i$-th \emph{cochain group} $C^i(K) = C^i(K, \Rspace)$ is the dual of the chain group, defined by $C^i(K):= \Hom(C_i(K), \Rspace)$, where $\Hom(C_i(K), \Rspace)$ denotes all homomorphisms of $C_i(K)$ into $\Rspace$. The coboundary operator, $\delta_i \colon C^i(K) \to C^{i+1}(K)$,  is defined as
        $ (\delta_if)([v_0,\ldots,v_{i+1}]) = \sum_{j=1}^{i+1} (-1)^{j}  f([v_0,\ldots,\hat v_j, \ldots, v_{i+1}]), $
        where $\hat v_j$ denotes that the vertex $v_j$ has been omitted.
        It satisfies the property   $\delta_i \delta_{i-1} = 0$ which implies that $\textrm{im} (\delta_{i-1}) \subset \textrm{ker}(\delta_i)$.
        The boundary operators, $\delta^\ast_i$,  are the adjoints of the  coboundary operators,
    
        $$ \cdots \ \ 
        C^{i+1}(K) 
        \underset{\delta_i^\ast}{\overset{\delta_i}{\leftrightarrows}}
        C^{i}(K) 
        \underset{\delta_{i-1}^\ast}{\overset{\delta_{i-1}}{\leftrightarrows}}
        C^{i-1}(K) 
        \ \ \cdots
        $$
        satisfying $(\delta_i a, b)_{C^{i+1}} = (a, \delta^\ast_i b)_{C^{i}}$ 
        for every $a \in C^i(K)$ and $b \in C^{i+1}(K)$, where $(\cdot, \cdot)_{C^i}$ denote the scalar product on the cochain group.
        
        Following \cite{HorakJost2013}, we define three combinatorial Laplace operators that operate on $C^i(K)$ (for the $i$-th dimension).
        Namely, the \emph{up Laplacian}, 
        $$\mathcal{L}_i^{\textrm{up}}(K) = \delta_i^\ast \delta_i,$$
        the \emph{down Laplacian},
        $\mathcal{L}_i^{\textrm{down}}(K) = \delta_{i-1} \delta_{i-1}^\ast,$
        and the \emph{Laplacian},
        $\mathcal{L}_i(K) =  \mathcal{L}_i^{\textrm{up}}(K) +  \mathcal{L}_i^{\textrm{down}}(K).$
        All three operators are self-adjoint, non-negative, compact and enjoy a collection of spectral properties, as detailed in~\cite{HorakJost2013}. We restrict our attention to the up Laplacians. 
    
    \para{Explicit expression for the up Laplacian.}
        To make the expression of up Laplacian explicit, we need to choose a scalar product on the coboundary vector spaces, which can be viewed in terms of weight functions~\cite{HorakJost2013}.
        In particular, the weight function $w$ is evaluated on the set of all simplices of $K$, 
        $w: \bigcup_{i=0}^{\dim K} S_i(K) \to \Rspace^+,$
        where the weight of a simplex $f$ is $w(f)$.
        Let $w_i\colon S_i(K) \to \Rspace^+.$
        Then $C^i(K)$ is the space of real-valued functions on $S_i(K)$, with inner product 
        $(a, b)_{C^i} := \sum_{f \in S_i(K)} w_i(f) a(f) b(f), $ 
        for every $a, b \in C^i(K)$.
        
        Choosing the natural bases, we identify each coboundary operator $\delta_p$ with an incidence matrix $D_p$.
        The \emph{incidence matrix} $D_p \in \mathbb R^{n_{p+1}} \times \mathbb R^{n_{p}}$ encodes which $p$-simplices are incident to which $(p+1)$ simplices in the complex, and is defined as 
        
        $$ D_p(i,j)  = \begin{cases}
        0 & \textrm{if $\sigma_j^{p}$ is not on the boundary of $\sigma_i^{p+1}$} \\
        1 & \textrm{if $\sigma_j^{p}$ is coherent with the induced orientation of $\sigma_i^{p+1}$} \\
        -1 & \textrm{if $\sigma_j^{p}$ is not coherent with the induced orientation of $\sigma_i^{p+1}$}
        \end{cases}
        $$
        Let $D_p^T$ be the transpose of $D_p$.
        Let $W_i$ be the diagonal matrix representing the scalar product on $C^i(K)$.
        The $i$-dimensional up Laplacian can then be expressed in the chosen bases, as the matrix
        
        $$\Lcal_{K,i} : = \mathcal{L}_i^{\textrm{up}}(K) = W^{-1}_{i}D_i^T W_{i+1} D_i. $$
        With this notation, $L = \Lcal_{K,0}$ is the graph Laplacian.
    
    \para{Effective resistance.}
        We quickly review the notation in~\cite{SpielmanSrivastava2011} regarding effective resistance.
        Let $G = (V,E,w)$ be a connected weighted undirected graph with $n$ vertices and $m$ edges and edge weights $w_e \in \Rspace^+$.
        $W$ is an $m \times m$ diagonal matrix with $W(e,e) = w_e$.
        Suppose the edges are oriented arbitrarily.
        Its graph Laplacian $L \in \Rspace^{n \times n}$ can be written as
        
        $$L = B^T W B, $$
        where $B \in \Rspace^{m \times n}$ is the signed edge-vertex incidence matrix, that is,
        
        $$
        B(i,j)  = \begin{cases}
        0 & \textrm{if vertex $j$ is not on the boundary of edge $i$}\\
        1 & \textrm{if $j$ is $i$'s head} \\
        -1 & \textrm{if $j$ is $i$'s tail}. 
        \end{cases}
        $$
        The effective resistance $R_e$ at an edge $e$ is the energy dissipation (potential difference) when a unit current is injected at one end and removed at the other end of $e$~\cite{SpielmanSrivastava2011}.
        Define the matrix 
        $R := B (L)^{+} B^T = B (B^T W B)^+ B^T,$ 
        where  $L^+$ is the Moore-Penrose pseudoinverse of $L$.
        The diagonal entry $R(e, e)$ of $R$, is the effective resistance $R_e$ across $e$.
        That is, $R_e = R(e,e)$.
        
        The above expression for $L$ is consistent with previous notation of up Laplacian, by setting $B = D_0$, $W = W_1$ for 
        $L =  \Lcal_{K,0} = W^{-1}_{0}D_0^T W_{1} D_0$.
        Suppose $W_0 = I$ (identity matrix), then $R$ could be expressed as 
        $R = D_0 (L)^+ D_0^T = D_0 (D_0^T W_{1} D_0)^+ D_0^T.$
    
    \para{Graph sparsification.}
        There are several different notions of approximation for graph sparsification, including the following based on spectral properties of the  associated graph Laplacian.
        We say  $H =(V,F,u)$ is an \emph{$\epsilon$-approximate sparse graph} of $G=(V,E,w)$ if $F \subset E$ and
        \begin{equation} 
        \label{eq:GraphSparse}
        (1-\epsilon) L_G \, \preceq \, L_H \, \preceq \, (1+\epsilon) L_G , 
        \end{equation}
        where $L_G$ and $L_H$ are the graph Laplacians of $G$ and $H$ respectively  and the inequalities are to be  understood in the sense of the semi-definite matrix ordering.
        That is, $\forall x \in \Rspace^n$, 
        $(1 - \epsilon) x^T L_G x  \,\leq \, x^T L_H x \,\leq \,(1+ \epsilon) x^T L_G x$.
	

\section{Sparsification of simplicial complexes}
\label{sec:algorithm}
    
    To prove the existence of an $\epsilon$-approximate sparse simplicial complex, we will follow the approach of \cite{SpielmanTeng2011} for the analogus problem for graphs.
    
    \para{Generalized  effective resistance for simplicial complexes.}
        To generalize effective resistance for simplices beyond dimension 1
        (i.e. edges), we consider the operator $R_{i} \colon C^{i} \to C^{i}$,
        defined by
        
        $$
        R_{i} = D_{i-1} (\mathcal{L}_{K, i-1})^\pinv D_{i-1}^T = D_{i-1} \left(
        W_{i-1}^{-1}D_{i-1}^T W_{i} D_{i-1} \right)^\pinv D_{i-1}^T. $$
        Specifically, setting $W_{i-1} = I$, we have
        
        $$
        R_{i} = D_{i-1} (\mathcal{L}_{K, i-1})^\pinv D_{i-1}^T
        = D_{i-1} \left( D_{i-1}^T W_{i} D_{i-1} \right)^\pinv D_{i-1}^T,
        $$
        which is the projection onto the image of $D_{i-1}$\footnote{For the rest of this section, we will assume $W_{i-1} = I$ in the simplicial complex $K$. Our results hold for any other choice of weights $W_{i-1}$ in dimension $i-1$ since for symmetric matrices $A, B$, $A \succeq B$ if and only if $DA \succeq DB$ for any positive definite diagonal matrix $D$.}
        The \emph{generalized effective resistance} on the $i$-dimensional simplex $f$, is defined to be the diagonal entry, $R_i(f,f)$.
        
        For  $i=1$, the generalized effective resistance reduces to the effective resistance on the graph  \cite{GhoshBoydSaberi2008}.
        That is, substituting $B = D_{0}$ and $W_0 = I$ in the notation from Section \ref{sec:background}, we have $R = R_1= D_0 (D_0^T W_{1} D_0)^+ D_0^T$.
        
        \begin{algorithm}[ht!]
        	\caption{$J= {\mathrm{\mathbf{Sparsify}}(K,i,q)}$} 
        	\label{alg:sparsify}
        	\KwData{A weighted, oriented simplicial complex $K$, a dimension $i$ (where $1 \leq i \leq \dime{K}$), and an integer $q$.  } 
        	\medskip
        	\KwResult{A weighted, oriented simplicial complex $J$ which is sparsified at dimension $i$, with equivalent $(i-1)$-skeleton to $K$ and $\dime{J} = i$. }
        	\medskip 
        	$J:= K^{(i-1)}$\\
        	Sample $q$ $i$-dimensional simplices independently with replacement according to the probability 
        	$$p_f = \frac{w(f) R_i(f,f)}{\sum_f w(f) R_i(f,f)},$$
        	and add sampled simplices to $J$ with weight  $w(f)/qp_f$. 
        	If a simplex is chosen more than once, the weights are summed. 
        \end{algorithm}
    
    \para{Sparsification algorithm.} 
        Algorithm \ref{alg:sparsify}, is a  natural generalization of the ${\mathrm{\mathbf{Sparsify}}}$ Algorithm given in \cite{SpielmanSrivastava2011}.
        The algorithm sparsifies a given simplicial complex $K$ at a fixed dimension $i$ (while ignoring all dimensions larger than $i$).
        The main idea is to include each $i$-simplex $f$ of $K$ in the sparsifier $J$ with probability proportional to its generalized effective resistance.
        Specifially, for a fixed dimension $i$, the algorithm chooses a random $i$-simplex $f$ of $K$ with probability $p_f$ (proportional to $w_fR_f$), and adds $f$ to $J$ with weight $w_f/qp_f$; 
        then $q$ samples are taken independently with replacement, while summing weights if a simplex is chosen more than once.
        The following theorem (Theorem \ref{thm:main}) shows that if $q$ is sufficiently large, the $(i-1)$-dimensional up Laplacians of $K$ and $J$ are close.
        
        \begin{theorem} 
        	\label{thm:main}
        	Let $K$ be a weighted, oriented simpicial complex, and $J= {\mathrm{\mathbf{Sparsify}}}(K,i,q)$ for some fixed $i$ (where $1 \leq i \leq \dime{K}$). 
        	Suppose $K$ and $J$ have $(i-1)$-th up Laplacians $\mathcal{L}_K := \mathcal{L}_{K,i-1}$ and $\mathcal{L}_J := \mathcal{L}_{J,i-1}$ respectively.
        	Let $n_{i-1}$ denote the number of $(i-1)$-simplices in $K$. 
        	Fix $\epsilon > 0$ (where $1/\sqrt{n_{i-1}} < \epsilon \leq 1$), and let $q =9 C^2 n_{i-1} \log n_{i-1} / \epsilon^2$, where $C$ is an absolute constant.  
        	If $n_{i-1}$ is sufficently large, then with probability at least $1/2$, 
        
        	\begin{equation} 
        	\label{eq:SpectralSparse} 
        	(1-\epsilon)\Lcal_K \,\preceq\, \Lcal_J \,\preceq \, (1+\epsilon)\Lcal_K,
        	\end{equation}
        	where the inequalities are to be understood in the sense of the semi-definite matrix ordering.
        	Equivalently, this means, $\forall x \in \Rspace^{n_{i-1}}$, 
        	$(1-\epsilon) x^T \Lcal_K x  \,\leq \, x^T \Lcal_J x  \,\leq \, 
        	(1+ \epsilon) x^T \Lcal_K x.$
        \end{theorem}
        
        \begin{proof} 
        	For simplicity in notation, let $\mathcal{L} =  \mathcal{L}_{K}:= \mathcal{L}_{K, i-1}$ and $\tilde{\mathcal{L}} =  \mathcal{L}_{J} 
        	:= \mathcal{L}_{J, i-1}$, with corresponding weight matrices denoted as $W_i$ and $\tilde{W}_i$ respectively.  
        	
        	Our proof follows  the proof of \cite[Theorem 1]{SpielmanSrivastava2011}. 
        	We consider the projection matrix $\Pi = W_i^{1/2} R_i W_i^{1/2}$. 
        	We also define the $n_i \times n_i$ nonnegative, diagonal matrix $S_i$ with entries 
        
        	$$
        	S_i(f,f) = \frac{ \tilde{w}_f }{w_f} = \frac{\textrm{\# times $f$ is sampled}}{q p_f},
        	$$
        	where the random entry $S_i(f,f)$ captures the number of $i$-simplices $f$ included in $J$ by ${\mathrm{\mathbf{Sparsify}}}$. 
        	The weight of an $i$-simplex $f$ in $J$ is $\tilde{w}_f = S_i(f,f) w_f$. 
        	Since $\tW_{i-1} = W_{i-1}S_{i-1} = W_{i-1}^{1/2} S_{i-1} W_{i-1}^{1/2}$ and $\tW_{i} = W_{i}S_{i} = W_{i}^{1/2} S_{i} W_{i}^{1/2}$, the $(i-1)$-dimensional up Laplacian of $J$ is therefore
        	
        	$$
        	\tLcal = \Lcal_{J, i-1} =  \tW_{i-1} D_{i-1}^T \tW_i  D_{i-1} 
        	= (W_{i-1}^{1/2} S_{i-1} W_{i-1}^{1/2}) D_{i-1}^T (W_{i}^{1/2} S_{i} W_{i}^{1/2}) D_{i-1}.
        	$$
        	Since $\Exp S_i = I$, $\Exp S_{i-1} = I$ and suppose $W_{i-1} = I$, therefore $\mathbb E \tilde{\mathcal{L}} = \mathcal{L}$. It is not difficult to show that if $S$ is a non-negative diagonal matrix such that 
        	
        	\begin{equation} \label{eq:EquivS}
        	\| \Pi S \Pi - \Pi \Pi \|_2 \leq \epsilon 
        	\end{equation}
        	then \eqref{eq:SpectralSparse} holds. 
        	But, $ \Pi S \Pi $ can be expressed as the average of symmetric, rank-one matrices. 
        	Now applying a result of  Rudelson and Vershynin \cite[Theorem 3.1]{RudelsonVershynin2007}, we have that 
        	$\mathbb E \| \Pi S \Pi - \Pi \Pi \|_2 \leq \frac{\epsilon}{2}$. By Markov's inequality this implies that \eqref{eq:EquivS} holds with probability at least $\frac{1}{2}$.
        \end{proof}
    
    \para{Computing the generalized effective resistance.}
        Spielman and Teng~\cite{SpielmanTeng2014} proposed a nearly linear time algorithm for solving symmetrical diagonally dominant (SDD) linear systems (that improved upon their previous results~\cite{SpielmanTeng2003}), which can be applied to graph sparsification.
        In particular, it has been proven that every weighted graph with $n$ vertices and $m$ edges has as an \emph{$\epsilon$-approximate sparse graph} with at most $O(n\cdot \log(n)/ \epsilon^2)$ edges and moreover, by sub-sampling the original graph with probabilities based on effective resistance, this graph can be found efficiently in $O(m \log(r)/\epsilon^2)$ time where $r$ is the ratio of largest to smallest edge weight \cite{SpielmanSrivastava2011}.
        Several recent SDD solvers based on low-stretch spanning trees improve the running time even further~\cite{KelnerOrecchiaSidford2013}. 
        The fastest known SDD solver proposed by Cohen et al. \cite{CohenKyngMiller2014} has $O(m\cdot \log^{1/2} n \log(1/\epsilon))$ time complexity for an $n \times n$ SDD matrix with $m$ non-zero entries.
        
        However, it should be noted, that while the graph Laplacian ($\mathcal{L}_{K, 0}$) is weakly SDD, the up Laplacians $\mathcal{L}_{K, i}$ for $i \geq 1$ are not diagonally dominant.
        Therefore, these fast SDD solvers cannot be used directly to compute generalized effective resistance. There has been some work on transforming non-SDD systems to or approximating non-SDD systems by an SDD system.
        In particular, using this approach, Boman et al.~\cite{BomanHendriksonVavasis2008} and Avron et al.~\cite{AvronChenShklarski2009} have proposed preconditioners for solving elliptic finite element systems in nearly linear time.
        However, more analysis is required before we can apply these approaches to speed up our sparsification algorithm.
        Solving linear systems in the $1$-dimensional up Laplacian has been studied by Cohen et al. for limited classes of complexes~\cite{CohenFasyMiller2014}.
        There is also a related line of work on spectral algorithms for $2$-dimensional  truss matrices initiated by Daitch and Spielman~\cite{DaitchSpielman2007}, although the numerical structures of such matrices are quite different.

        Alternatively, we can look at sparsification using generalized effective resistance as a form of leverage score sampling, where rows or columns of a matrix are sampled with probabilities proportional to their relative size (i.e.,~norm).
        To see this relation, define $\Phi = W_{i}^{1/2}D_{i-1}$ to be a scaled incidence matrix. Suppose $W_{i-1} = I$ and the projection matrix $\Pi = W_{i}^{1/2}R_{i}W_{i}^{1/2}$ is defined as before.
        For an $i$-dimensional simplex $f$ of $K$, the corresponding diagonal entry of $\Pi$ is given by $w(f)R_i(f,f)$.
        With a little algebraic manipulation, we can show that this is precisely the leverage score of the row of $\Phi$ corresponding to $f$, under the $L_2$ norm. Thus the probability of sampling $f$ is given by the normalized leverage score of the corresponding row of $\Phi$.
        
        Although the computation of exact leverage scores has the same time complexity as the computation of generalized effective resistance, there has been some work in fast approximation of leverage scores \cite{DrineasMalikMahoneyWoodruff2012}.
        We may be able to use a similar approach to approximate generalized effective resistance and improve the runtime efficiency of our sparsification algorithm.
        However, a careful analysis is required to determine how the approximation of generalized effective resistance affects sparsification bounds from Equation (\ref{eq:SpectralSparse}).


\section{Generalized Cheeger inequalities for simplicial complexes} 
\label{sec:Cheeger}

    \para{Cheeger constant and inequality  for graphs.}
        The Cheeger constant for an unweighted graph $G=(V,E)$ is given by~\cite{GundertSzedlak2015} 
        
        \begin{equation} 
        \label{e:ChUnGr}
        h(G) := \min_{\varnothing \subsetneq A  \subsetneq V} \ \frac{|V| \ | E (A,V \setminus A)|}{|A| \ |V \setminus A|},
        \end{equation}
        where  $E(A, V \setminus A)$ is the set of edges that connect $A\subset V$ to $(V \setminus A) \subset V$.
        For a weighted graph, $G = (V,E,w)$,  the Cheeger constant is typically generalized to
        
        \begin{equation} \label{e:ChWeGr}
        h(G) := \min_{\varnothing \subsetneq A \subsetneq V}  \ \frac{|V|}{|A| \ |V \setminus A|}\sum_{(i,j) \in E(A,V \setminus A)} w_{ij}. 
        \end{equation}
        The Cheeger inequality for graphs takes the form: 
        $c \cdot \lambda_1(L_G) \leq  h(G) \leq C \cdot  \sqrt{\lambda_1(L_G)}$, 
        where $\lambda_1$ is the first non-trivial eigenvalue of a graph Laplacian and $c$ and $C$ are constants which depend on the choice of definition for the Cheeger constant and graph Laplaican; see, {\it e.g.}, \cite[Chapter~2]{Chung1997}.
        Using the variational formulation for eigenvalues and a suitable test function, it isn't difficult to prove that for the Cheeger constant defined in \eqref{e:ChWeGr} and the weighted graph Laplacian, the first inequality (a lower bound for the Cheeger constant) is given by $\frac{1}{2} \cdot \lambda_1(L_G) \leq h(G)$.
        In the following, we prove an analogous inequality for weighted simplicial complexes, which we refer to as a generalized Cheeger inequality.
        This inequality gives a lower bound on the Cheeger constant; the upper bound isn't possible for weighted simplicial complexes by the argument of Gundert and Szedl\'ak~\cite[p.5]{GundertSzedlak2015}.
    
    \para{Generalized Cheeger inequality for simplicial complexes of Gundert and Szedl\'ak. }
        We first recall the generalized Cheeger inequality for simplicial complexes of Gundert and Szedl\'ak~\cite{GundertSzedlak2015}, which only 
        For a $k$-dimensional simplicial complex $K$, its \emph{$k$-dimensional completion} is defined to be 
        
        \begin{align*}
        \bar{K}:= K \bigcup \{\tau^* \in {V \choose k+1} \mid 
        (\tau^* \setminus {v}) \in X,  \forall v \in \tau^*\}.
        \end{align*} 
        When $K$ has a complete $(k-1)$-skeleton, $\bar{K}$ is the complete $k$-dimensional complex.
        The generalized Cheeger constant for \emph{unweighted} simplicial complexes is defined to be  
        
        \begin{equation} \label{e:ChUnSC}
        h(K) := \min_{\substack{V = \bigsqcup_{i=0}^k A_i\\ A_i \neq \varnothing}}
        \frac{|V||F(A_0, A_1, \ldots, A_k)|}{|F^*(A_0, A_1, \ldots, A_k)|},
        \end{equation}
        where $F(A_0,A_1,\ldots,A_k)$ and $F^*(A_0,A_1,\ldots,A_k)$ are the sets
        of all $k$-simplices of $K$ and $\overline{K}$, respectively, with one node in
        $A_i$ for all $0 \leq i \leq k$.
        
        \begin{theorem}[{\cite[Theorem~2]{GundertSzedlak2015}}] \label{thm:GS2015}
        	If $\lambda_1(\Lcal_K)$ is the first non-trivial eigenvalue of the $k$-th up-Laplacian and if every $(k-1)$-face is contained in at most $C^*$ $k$-face of $K$, then 
        	$$
        	\frac{|V|}{(k+1) \ C^*} \cdot \lambda_1(\Lcal_K)  \leq  h(K). 
        	$$
        \end{theorem}
    
    \para{Remark.}
        Recall that the Cheeger inequality for graphs includes an upper bound of the Cheerger constant $h(G)$ in terms of $\lambda_1(L_G)$.
        However, as pointed out by Gundert and Szedl\'ak, $\lambda_1(\Lcal_K) = 0$ does not imply $h(K)=0$~\cite{ParzanchevskiRosenthalTessler2016}, a higher-dimensional analogue of this upper bound of the form 
        $h(K) \leq C \cdot \lambda_1(\Lcal_K)^{\frac{1}{m}}$ is not possible~\cite{GundertSzedlak2015}. 
        We also remark that an alternative Cheeger inequality is given in \cite{ParzanchevskiRosenthalTessler2016}.
    
    \para{A generalized Cheeger constant for weighted simplicial complexes.}
        In analogy to the generalization of the unweighted Cheeger constant in Equation~\eqref{e:ChUnGr} to the weighted Cheeger constant in Equation \eqref{e:ChWeGr}, we define the generalized Cheeger constant for weighted simplicial complexes by 
        
        \begin{equation} \label{e:ChWeSC}
        h(K) := \min_{\substack{V = \bigsqcup_{i=0}^k A_i\\ A_i \neq
        		\varnothing}}  \ 
        \frac{|V|}{|F^*(A_0, A_1, \ldots, A_k)|}  \ 
        \sum_{X \in F(A_0,A_1,\ldots,A_k)}  w_k(X).
        \end{equation} 
        Observe that Equation~\eqref{e:ChWeSC} agrees with   Equation~\eqref{e:ChUnSC} in the case when all weights are unity. The following result can be proved analogously to Theorem \ref{thm:GS2015}. 
        
        \begin{proposition} 
        	\label{prop:newCheeger}
        	If $\lambda_1(\Lcal_K)$ is the first non-trivial eigenvalue of the $k$-th weighted up-Laplacian and if every $(k-1)$-face is contained in at most $C^*$ $k$-face of $K$, then 
        	$$
        	\frac{|V|}{(k+1) \ C^*} \cdot \lambda_1(\Lcal_K)  \leq  h(K). 
        	$$
        \end{proposition}
        \noindent Proposition~\ref{prop:newCheeger} can be proven using a slight modification of the arguments in \cite{GundertSzedlak2015} by adapting weights in the definition of the generalized Cheeger constant \eqref{e:ChWeSC}. 
        
        The  following result now follows from combining Theorem \ref{thm:main} and Proposition~\ref{prop:newCheeger}. 
        
        \begin{corollary} In the  setting as Theorem \ref{thm:main} and Proposition~\ref{prop:newCheeger}, we have with probability $\frac{1}{2}$
        
        	$$
        	\frac{|V|} {(k+1) \ C^*}  (1-\epsilon) \cdot  \lambda_1(\Lcal_K) \leq 
        	\frac{|V|} {(k+1) \ C^*}  \cdot  \lambda_1(\Lcal_J) \leq h(J). 
        	$$
        \end{corollary}
        \noindent
        Thus, the Cheeger constant of the sparsified simplicial complex, $J$, is bounded below by a multiplicative factor of the first nontrivial eigenvalue of the up Laplacian for the original complex, $K$.


\section{Numerical experiments} 
\label{sec:experiments}

    In Section~\ref{sec:preservation}, we conduct numerical experiments to illustrate the inequalities bounding the spectrum of the up Laplacian of the sparsified simplicial complex, proven in Theorem \ref{thm:main}.
    In Section~\ref{subsec:spectral-clustering} we extend a well-known graph spectral clustering method to simplicial complexes. We show that the clusters obtained for sparsified simplicial complexes are similar to those of the original simplicial complex. We also present the analogous results for graph sparsification  to serve as a comparison.
    In Section~\ref{subsec:label-propagation}, we show similar results for label propagation before and after sparsification.
    The na\"{i}ve implementation of spectral clustering is quadratic in the number of simplices; while label propagation is cubic.
    While we could take advantage of sparse matrix methods (see Appendix \ref{sec:app-complexity} for details), our proposed sparsification method could  further improves these computational complexity estimates.


\subsection{Preservation of the spectrum of the up Laplacian} 
\label{sec:preservation}

    \para{Experimental set up.}
        In the setting of graph sparsification~\cite{SpielmanSrivastava2011}, we recall that if a graph $H$ is an $\varepsilon$-approximation of a graph $G$, $n$ is the number of vertices in $H$ and $G$,  then we have the following inequality,
        
        \begin{equation}
        \label{eq:graph-inequality}
        (1 - \varepsilon) x^T L_G x \,\leq \, x^T L_H x \, \leq \,(1+ \varepsilon) x^T L_G x, 
        \qquad \qquad \qquad  \forall x \in \Rspace^n. 
        \end{equation}
        Subtracting $x^T L_G x$ from all terms in this inequality, we obtain
        
        \begin{equation*}
        -\varepsilon x^T L_G x \, \leq \, x^T (L_H - L_G) x \, \leq \, \varepsilon x^T (L_G) x, 
        \qquad \qquad \qquad  \forall x \in \Rspace^n. 
        \end{equation*}
        Let $\lambda_{max}(L_G)$, $\lambda_{max}(L_H)$ and  
        $\lambda_{max}(L_H - L_G)$ be the maximum eigenvalues of $L_G$ and $L_H$ and $L_H- L_G$ respectively.
        Also, let $\lambda_{min}(L_G)$ be the minimum eigenvalue of $L_G$.  
        With some algebraic manipulations, we obtain on the right hand side,
        
        \begin{align*}
        \lambda_{max}(L_H - L_G)  = \max_{||x|| = 1} x^T (L_H - L_G) x  \leq  \varepsilon \max_{||x|| = 1} x^T (L_G) x  = \varepsilon \lambda_{max}(L_G).
        \end{align*}
        Similarly, on the left hand side, we obtain
        
        \begin{align*}
        0 =  - \varepsilon \lambda_{min}(L_G)  
        = - \varepsilon \min_{||x|| = 1} x^T L_G x 
        & = \max_{||x|| = 1} - \varepsilon x^T L_G x   \\
        &  \leq  \max_{||x|| = 1} x^T (L_H - L_G) x = \lambda_{max}(L_H - L_G).
        \end{align*}
        Together we have the inequality 
        
        \begin{equation}
        \label{eq:graph-eigenvalue-new}
        0 
        \ \leq \  \lambda_{max}(L_H - L_G) 
        \ \leq \ \varepsilon \lambda_{max}(L_G).  
        \end{equation}
        We can obtain the analogous inequality in the setting of simplicial complex sparsification. 
        Let $J$ be a sparsified version of $K$ following the setting of Theorem~\ref{thm:main}.
        Suppose for a fixed dimension $i$ (where $1 \leq i \leq \dime{K}$), $K$ and $J$ have $(i-1)$-th up Laplacians $\mathcal{L}_K := \mathcal{L}_{K,i-1}$ and $\mathcal{L}_J := \mathcal{L}_{J,i-1}$ respectively, we have,
        
        \begin{equation}
        \label{eq:sc-inequality}
        (1 - \varepsilon) x^T \mathcal{L}_K x \leq  x^T \mathcal{L}_J x \leq (1+ \varepsilon) x^T \mathcal{L}_K x, 
        \qquad \qquad \qquad  \forall x \in \Rspace^{n_{i-1}}. 
        \end{equation}
        A similar argument leads to the following inequality, 
        
        \begin{equation}
        \label{eq:sc-eigenvalue-new}
        0 \ \leq \  \lambda_{max}(\mathcal{L}_J - \mathcal{L}_K) 
        \ \leq \ \varepsilon\lambda_{max}(\mathcal{L}_K).  
        \end{equation}
        Notice that inequality~\eqref{eq:graph-eigenvalue-new} is a special case of the inequality~\eqref{eq:sc-eigenvalue-new}.
    
    \para{Preservation of the spectrum of the sparsified graph Laplacian.}
        To demonstrate how the spectrum of the graph Laplacian is preserved during graph sparsification, we set up the following experiment.
        Consider a complete graph $G$ with $n_0 = 40$ vertices and $n_1 = 780$ edges.
        We run multiple sparsification processes on this graph $G$ and study the convergence behavior based on the inequality in \eqref{eq:graph-inequality}.
        For each sparsification process, we use a sequence of sample sizes, ranging between $10$ and $2n_1$. For each sample size $q$, we set $\varepsilon =\sqrt{n_0 \log n_0/q} $ by assuming that $9C^2 = 1$ in the hypothesis of Theorem~\ref{thm:main}.
        As $q$ varies, we correspondingly obtain a sequence of varying $\varepsilon$ values.
        
        In particular, we run 25 simulations on $G$.
        For each simulation, we fix a unit vector $x$ uniformly randomly sampled from $\Sspace^{n_0}$, and perform 25 instances of experiments.
        For each instance, we apply our sparsification procedure to generate the convergence plot using the list of fixed sample sizes $q$ and their corresponding $\varepsilon$'s.
        Specifically, for each sample size, we obtain a sparse graph $H$ and compute $x^T L_{H} x$ and $\lambda_{max}(L_{H} - L_{G})$; and we observe the convergence behavior of these quantities as the sample size increases.
        
        In Figure~\ref{fig:graph-results}(a), we show the convergence behavior based on the inequality in \eqref{eq:graph-inequality}.
        For a single simulation, we compute the point-wise average of $x^T L_H x$ across the $25$ instances, and plot these values as function of the sample size $q$, which gives rise to a single convergence curve in aqua.
        Then we compute the point-wise average of the aqua curves across all simulations, producing the red curve.
        Since each simulation (for a fixed $x$) has a different upper bound curve $(1-\varepsilon)x^T L_G x$ and lower bound curve $(1+\varepsilon)x^T L_G x$ respectively (not shown here), the point-wise average of the upper and lower bound curves across all simulations is plotted in blue.
        We observe that on average, these curves respect the inequality~\eqref{eq:graph-inequality}, that is, the red curve is nested within its approximated theoretical upper and lower bounds in blue. 
        
        In Figure~\ref{fig:graph-results}(b), we illustrate the theoretical upper and lower bounds for $\lambda_{max}(L_H - L_G)$ given in inequality~\eqref{eq:graph-eigenvalue-new} as the sample size $q$ increases.
        In particular, we run a single simulation with $25$ instances, computing $\lambda_{max}(L_{H} - L_{G})$.
        Each instance gives us a convergence curve shown in aqua. We compare the point-wise average of $\lambda_{max}(L_H - L_G)$ (in red) with its (approximated) theoretical upper bound in blue and lower bound ({\it i.e.},~$0$, the x-axis).
        On average, the experimental results respect the inequality~\eqref{eq:graph-eigenvalue-new}.  
        Figure~\ref{fig:scaling}(a) illustrates how the number of edges scale with the increasing number of samples across all instances.
        
        \begin{figure}[t!]
        	\begin{center}
        		\begin{tabular}{cc}
        			\includegraphics[width=0.49\linewidth]{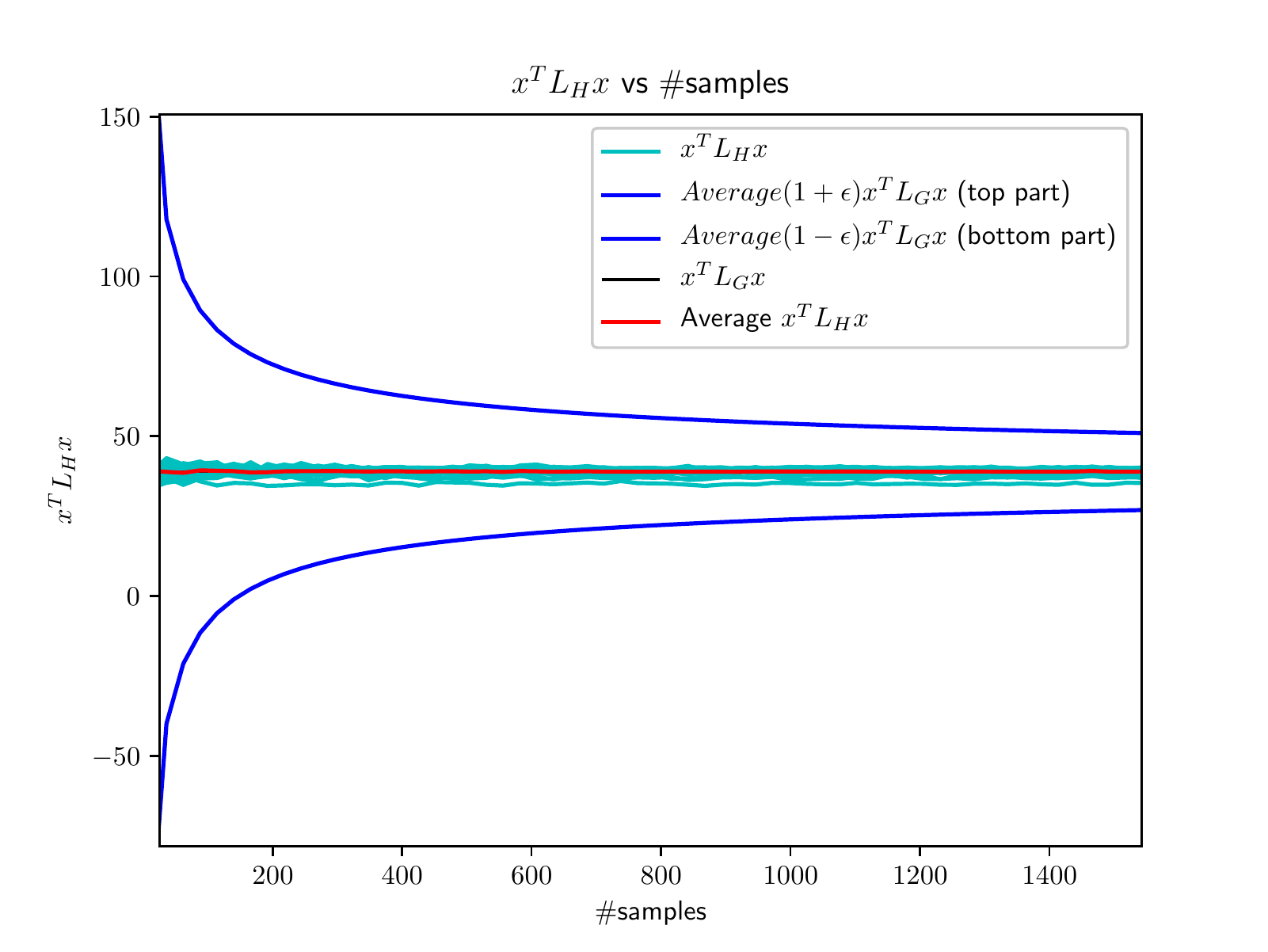} &
        			\includegraphics[width=0.49\linewidth]{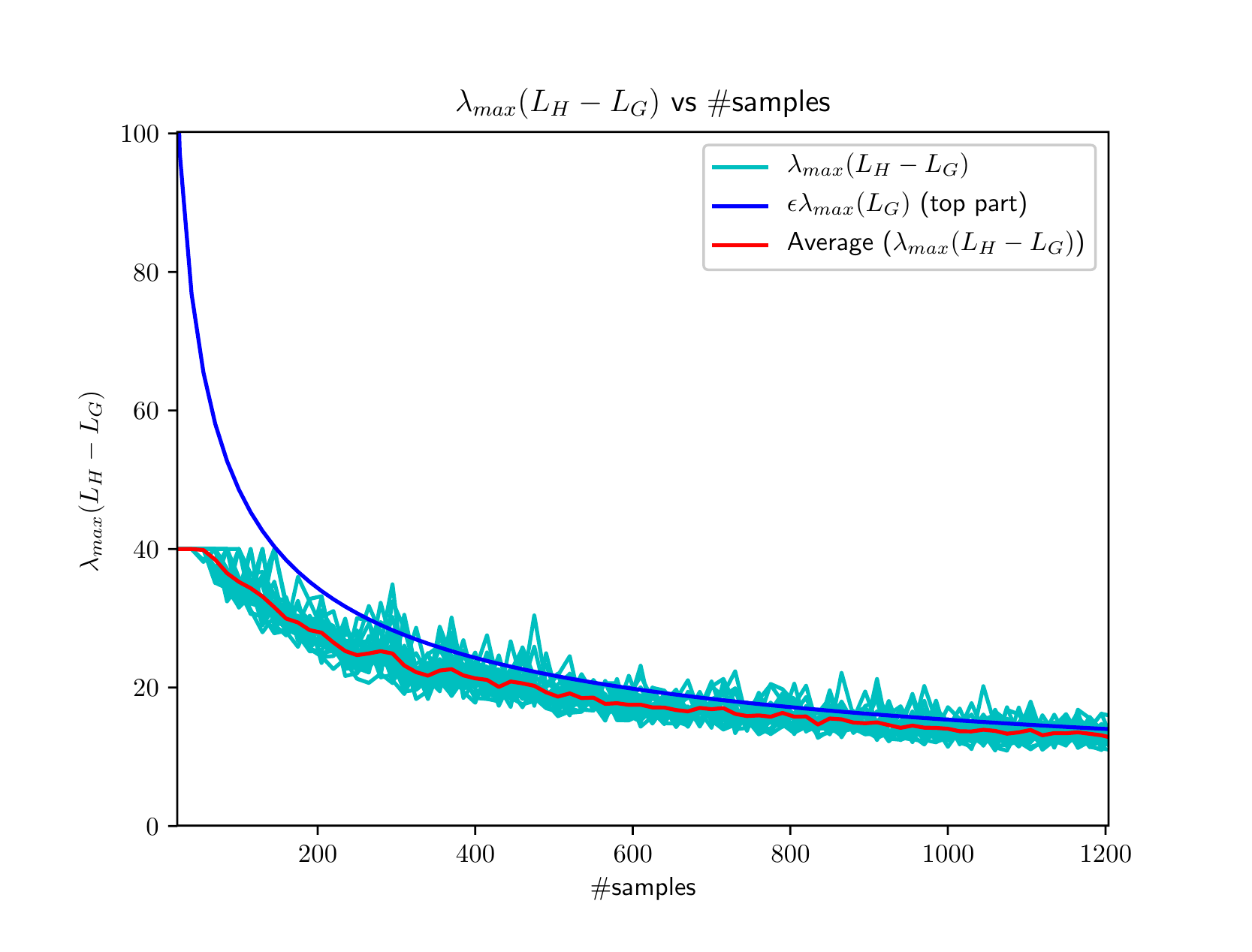} \vspace{-4mm}\\
        			{\bf (a)} & {\bf (b)}
        		\end{tabular}
        		\vspace{-2mm}
        		\caption{The results of a numerical experiment illustrating inequalities which control the spectrum of sparsified graph Laplaicans. 
        			{\bf (a)} For an ensemble of vectors, $x\in \mathbb S^{n_0}$, and sparsified graphs, $H$, we plot the terms in inequality \eqref{eq:graph-inequality}. 
        			{\bf (b)} For an ensemble of sparsified graphs, $H$, we plot the terms in the inequality \eqref{eq:graph-eigenvalue-new}. 
        		}
        		\label{fig:graph-results}
        	\end{center}
        \end{figure}
    
    \para{Preservation of the spectrum of the up Laplacian for a sparsified simplicial complex.}
        To demonstrate that the spectrum of the up Laplacian is preserved during the sparsification of a simplicial complex, we set up a similar experiment.
        We start with a $2$-dimensional simplicial complex, $K$, that contains all edges and triangles on $n_0 = 40$ vertices (with $n_1 = 780$ edges and $n_2 =  9880$ faces.) and a sequence of fixed sample sizes q.
        For each sample size $q$, we solve for $\varepsilon =\sqrt{n_1 \log n_1/q} $ assuming that $9C^2 = 1$ in the hypothesis of Theorem~\ref{thm:main}, to get the corresponding sequence of $\varepsilon$ values.
        With the simplicial complex $K$ and the sequence of sample sizes fixed, we run 25 simulations, each simulation consisting 25 instances and a fixed randomly sampled unit vector $x$ as described previously; only this time, we sparsify the faces of the simplicial complex by applying Algorithm~\ref{alg:sparsify} with $i=2$.
        In Figure~\ref{fig:sc-results}, we plot the terms in inequalities describing the spectrum for these sparsified simplicial complexes.
        
        In Figure~\ref{fig:sc-results}(a), following the same procedure as for graph sparsification, we obtain a plot that respects the inequality~\eqref{eq:sc-inequality}.
        The curves in aqua show the point-wise averages of $x^T \mathcal{L}_{J} x$ across all instances in a single simulation, whereas the red curve represents point-wise average across all instances and all simulations.
        Since the random vector $x$ is resampled for each simulation, the upper and lower bound curves are different for every simulation.
        In Figure~\ref{fig:sc-results}(a) we plot their point-wise average across all simulations as the upper and lower bound curves in blue.
        
        In Figure~\ref{fig:sc-results}(b), to illustrate inequality~\eqref{eq:sc-eigenvalue-new}, we run a single simulation with $25$ instances. Each instance gives us a sequence of $\lambda_{max}(\mathcal{L}_J - \mathcal{L}_K)$ values as function of sample size. We plot them as curves in aqua.
        We compare the point-wise averages of $\lambda_{max}(\mathcal{L}_J - \mathcal{L}_K)$ (in red) with its (approximated) theoretical upper and lower bounds in blue.
        Figure~\ref{fig:scaling}(b) shows how the number of faces scales with increasing number of samples across all instances.
        
        \begin{figure}[t!]
        	\begin{center}
        		\begin{tabular}{cc}
        			\includegraphics[width=0.49\linewidth]{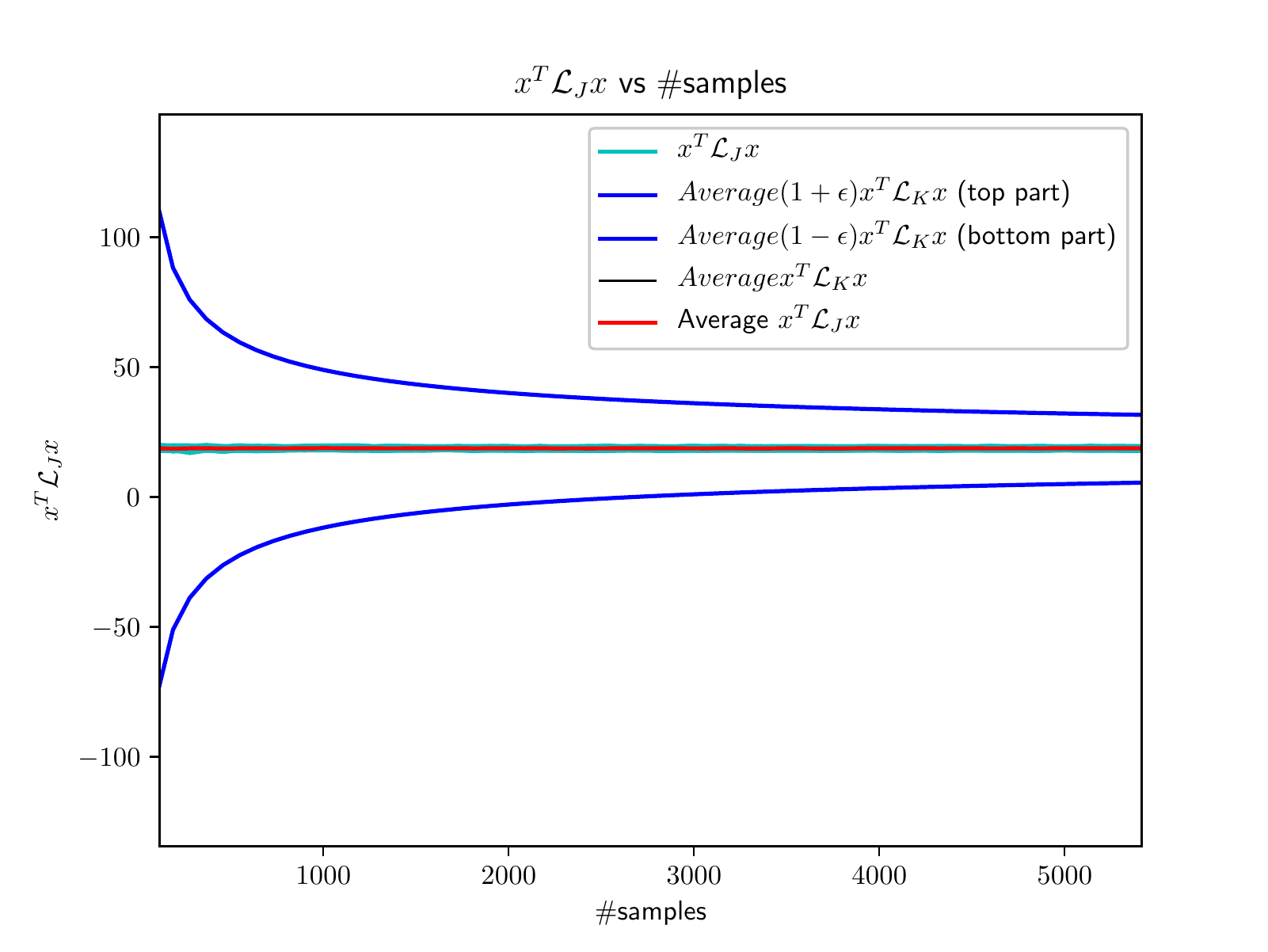} &
        			\includegraphics[width=0.49\linewidth]{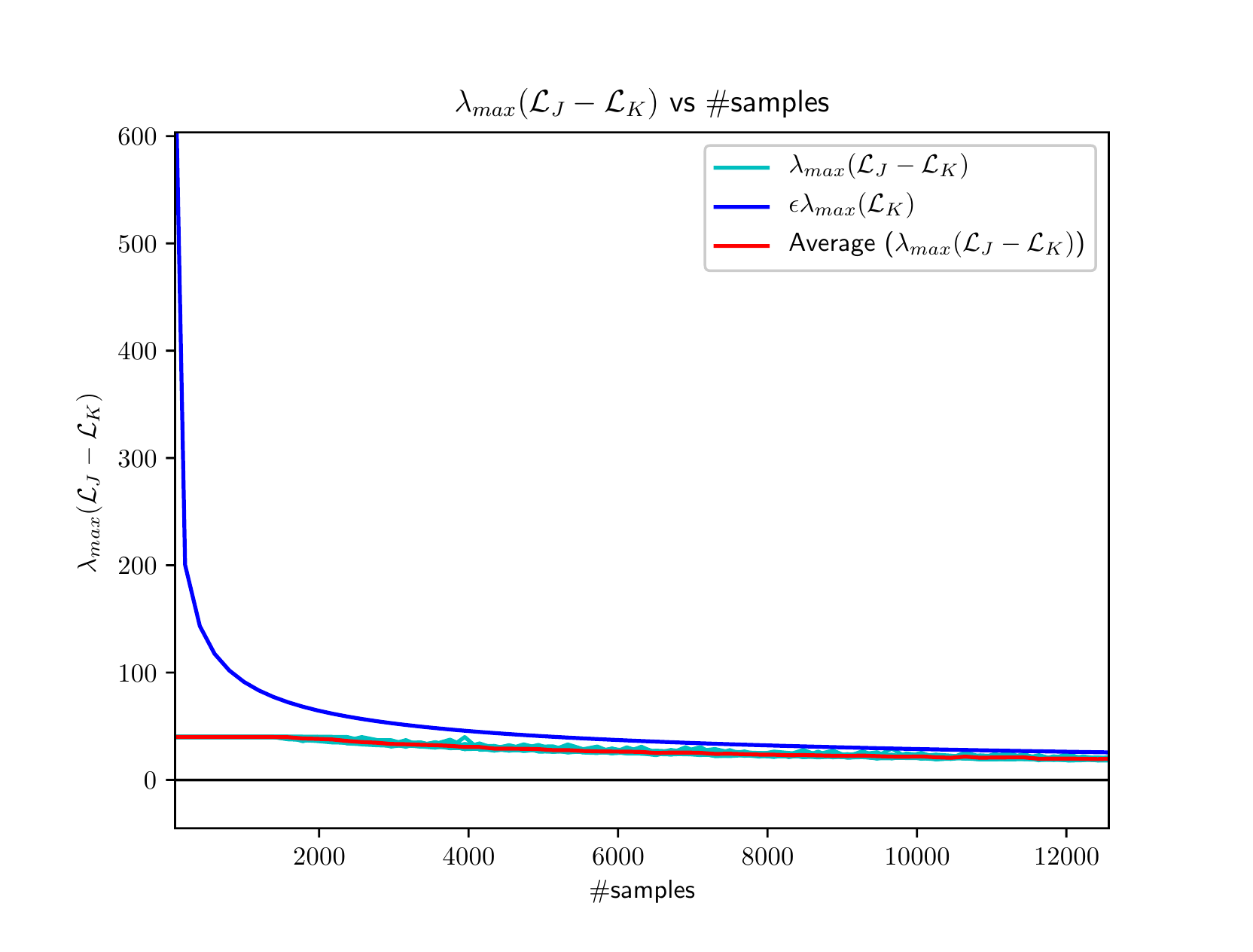} \vspace{-4mm}\\
        			{\bf (a)} & {\bf (b)}
        		\end{tabular}
        		\vspace{-2mm}
        		\caption{ The results of a numerical experiment illustrating inequalities which control the spectrum of the up Laplacian for sparsified simplicial complexes.
        			{\bf (a)} For an ensemble of vectors, $x\in \mathbb S^{n_1}$, and sparsified simplicial complexes, $J$, we plot the terms in inequality \eqref{eq:sc-inequality}. 
        			{\bf (b)} For an ensemble of sparsified simplicial complexes, $J$, we plot the terms in the inequality \eqref{eq:sc-eigenvalue-new}. }
        		\label{fig:sc-results}
        	\end{center}
        \end{figure}
        
        \begin{figure}[t!]
        	\begin{center}
        		\begin{tabular}{cc}
        			\includegraphics[width=0.49\linewidth]{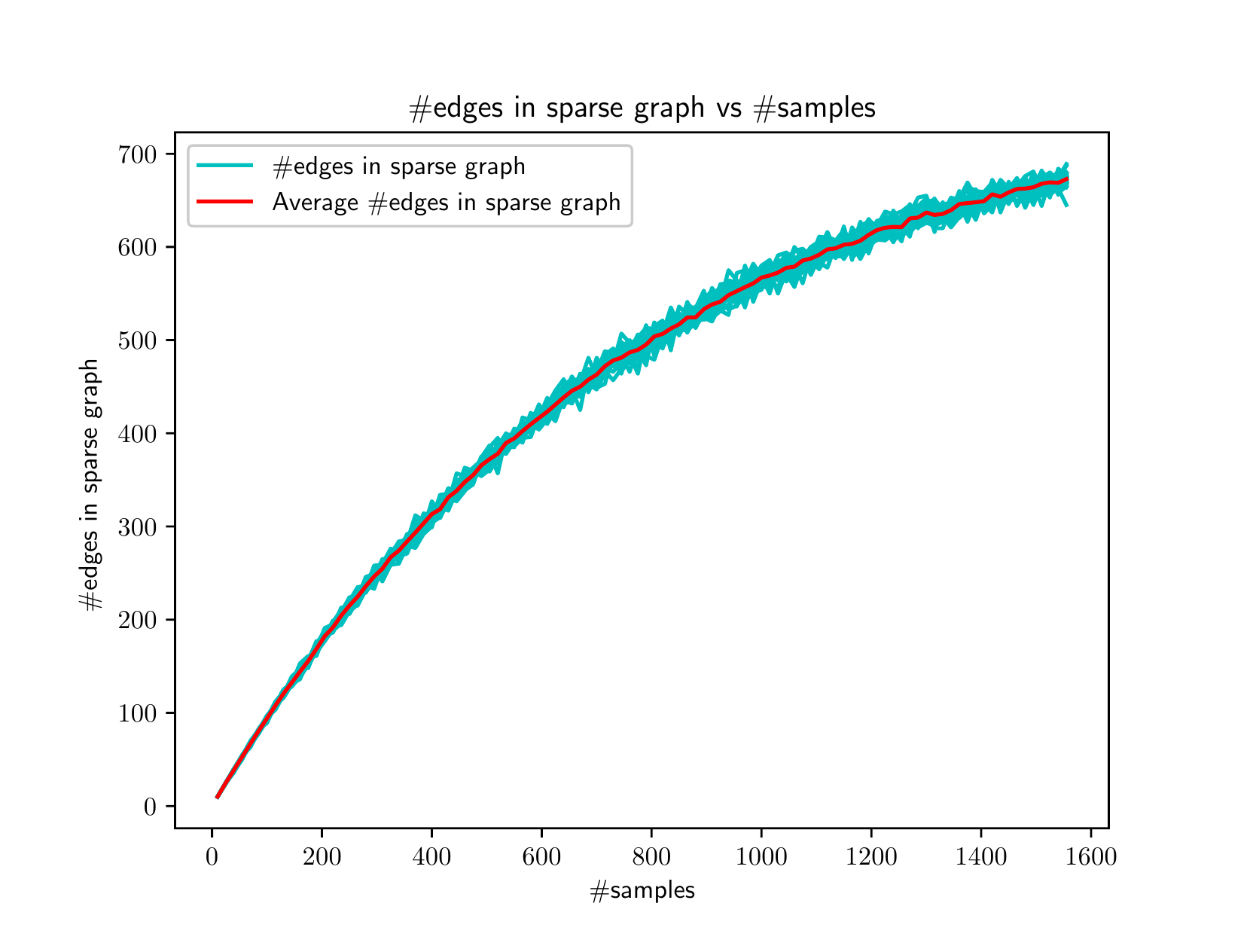} &
        			\includegraphics[width=0.49\linewidth]{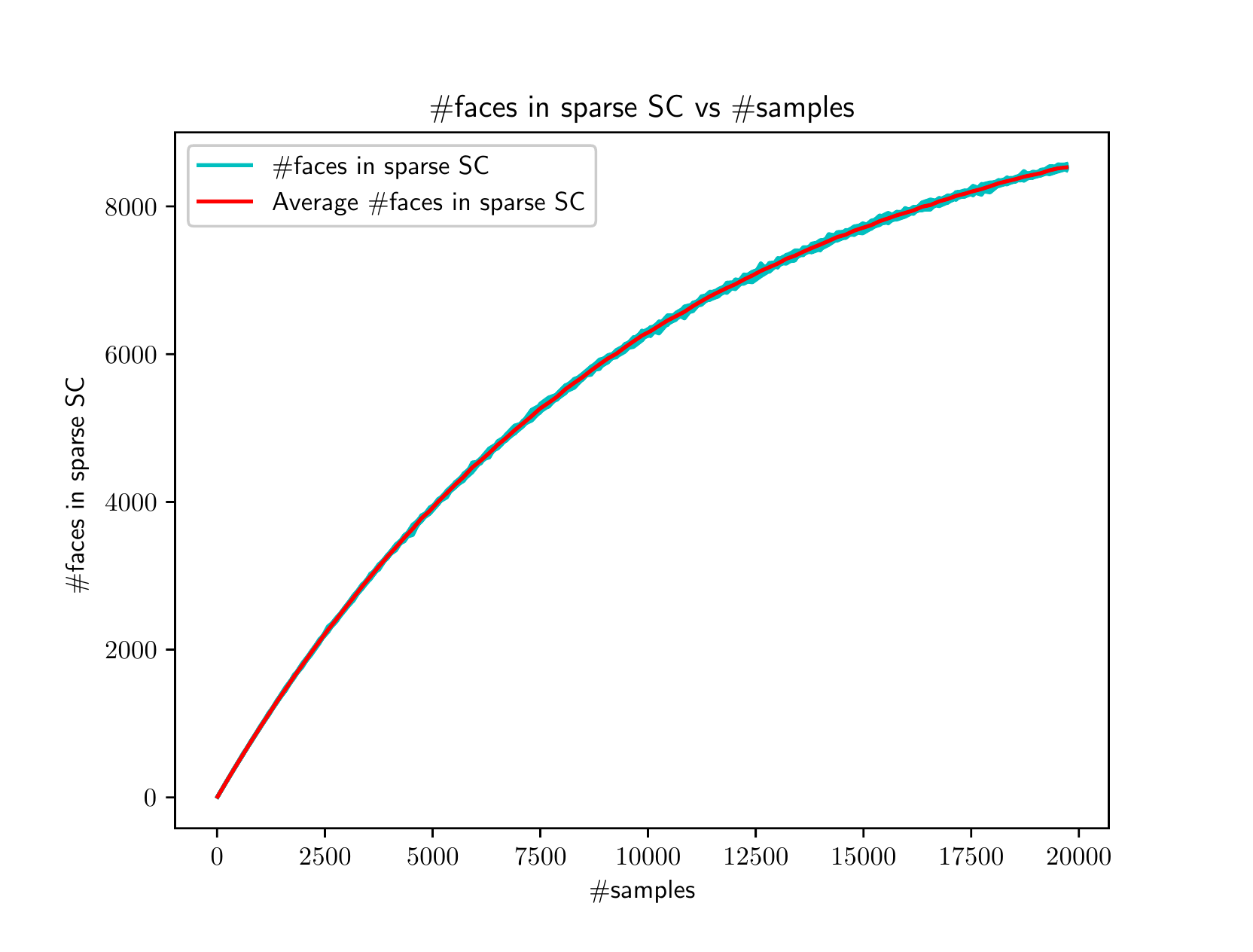} \vspace{-4mm}\\
        			{\bf (a)} & {\bf (b)}
        		\end{tabular}
        		\vspace{-2mm}
        		\caption{Figures illustrating how {\bf (a)} the number of edges in the case of graph sparsification and  {\bf (b)} the number of faces/triangles in the case of simplicial complex sparsification  vary with increasing sample size.}
        		\label{fig:scaling}
        	\end{center}
        \end{figure}
    

\subsection{Spectral clustering} 
\label{subsec:spectral-clustering}

    Spectral clustering can be considered as a class of algorithms with many variations. Here, we apply spectral clustering to simplicial complexes before and after sparsification.
    We demonstrate via numerical experiments, that preserving the structure of the up Laplacian via sparsification also preserves the results of spectral clustering on simplicial complexes.
    
    \para{Datasets.}
        We consider a graph that contains two complete subgraphs with $20$ vertices (and $190$ edges) each, which are connected by $64 = 8 \times 8$ edges spanning across the two subgraphs. We refer to this graph, $G$, as the \emph{dumbbell graph}; it has $n_0 = 40$ vertices and $n_1 = 444$ edges.
        All edge weights are set to be $1$. To compute the sparsified graph, the number of samples, $q$, is set to be $0.5 n_1$.
        
        Similarly, we consider a simplicial complex that contains two complete sub-complexes with $10$ vertices, $45$ edges and $120$ triangles each.
        The two sub-complexes are connected by $16$ cross edges and $48$ cross triangles so that the simplicial complex is made up of $n_0 = 20$ vertices, $n_1 = 106$ edges and $n_2 = 288$ triangles.
        We refer to this simplicial complex, $K$, as the \emph{dumbbell complex}.
        The weights on all edges and triangles are set to be $1$.
        To compute the sparsified simplicial complex, the number of samples, $q$, is set to be $0.75 n_2$.
    
    \para{Spectral clustering algorithm for graphs.}
        We use the Ng-Jordan-Weiss algorithm~\cite{NgJordanWeiss2001} detailed below to perform spectral clustering of graphs.
        Let $n_0$ be the number of vertices in a graph.
        Recall the \emph{affinity matrix} $A\in \Rspace^{n_0 \times n_0}$ is a matrix where $A_{ij}$ $(\geq 0)$ captures the affinity (i.e.~measure of similarity) between vertex $i$ and vertex $j$.
        In our setting, $A_{ij}$ corresponds to the weight of edge $e_{ij}$ in the diagonal edge weight matrix $W_1$.
        The spectral clustering algorithm in~\cite{NgJordanWeiss2001} can be summarized as follows:
        \begin{enumerate} 
        	\item Compute the diagonal matrix $\Delta \in \Rspace^{n_0 \times n_0}$  with diagonal elements $\Delta_{ii}$ being the sum of $A$'s $i$-th row, that is, $\Delta_{ii} = \sum_j{A_{ij}}$.
        	\item Construct the matrix  $M = \Delta^{-1/2}  A \Delta^{-1/2}$. 
        	\item Find $u_1, u_2, \cdots, u_k$, the eigenvectors of $M$ corresponding to the $k$ largest eigenvalues (chosen to be orthogonal to each other in the case of repeated eigenvalues), and form the matrix $X = [u_1 u_2 \cdots u_k] \in \Rspace^{n_0 \times k}$ by stacking the eigenvectors in columns. 
        	\item Form the matrix $Y$ from $X$ by re-normalizing each of $X$'s rows to have unit length, that is, $Y_{ij} = X_{ij}/\left( \sum_j X^2_{ij}\right)^{1/2}$. 
        	\item Treating each row of $Y$ as a point in $\Rspace^k$, cluster them into $k$ clusters via the $k$-means algorithm. 
        	\item Finally, assign the original vertex $v_i$ to cluster $j$ if and only if row $i$ of the matrix $Y$ is assigned to cluster $j$.
        \end{enumerate}
        The graph Laplacian can be written as $L = \Delta - A$.
        Furthermore $M = I - L_N$, where $L_N = \Delta^{-1/2} L \Delta^{-1/2}$ is referred to as the normalized graph Laplacian.
        In the case of a binary graph (where edge weights are either $0$ or $1$), the affinity matrix $A$ equals the vertex-vertex adjacency matrix; and $\Delta$ is the degree matrix with diagonal elements $\Delta_{ii}$ being the number of edges incident on vertex $v_i$.
        
        To demonstrate the utility of the sparsification, we illustrate the spectral clustering results before and after graph sparsification in Figure~\ref{fig:clustering-graphs-sc} (a)-(b).
        Since graph sparsification preserves the spectral properties of graph Laplacian, we expect it to also preserve (to some extent) the results of spectral methods, such as spectral clustering.
        
        \begin{figure}[!ht]
        	\vspace{-2mm}
        	\begin{center}
        		\begin{tabular}{cc}
        			\includegraphics[width=0.4\linewidth]{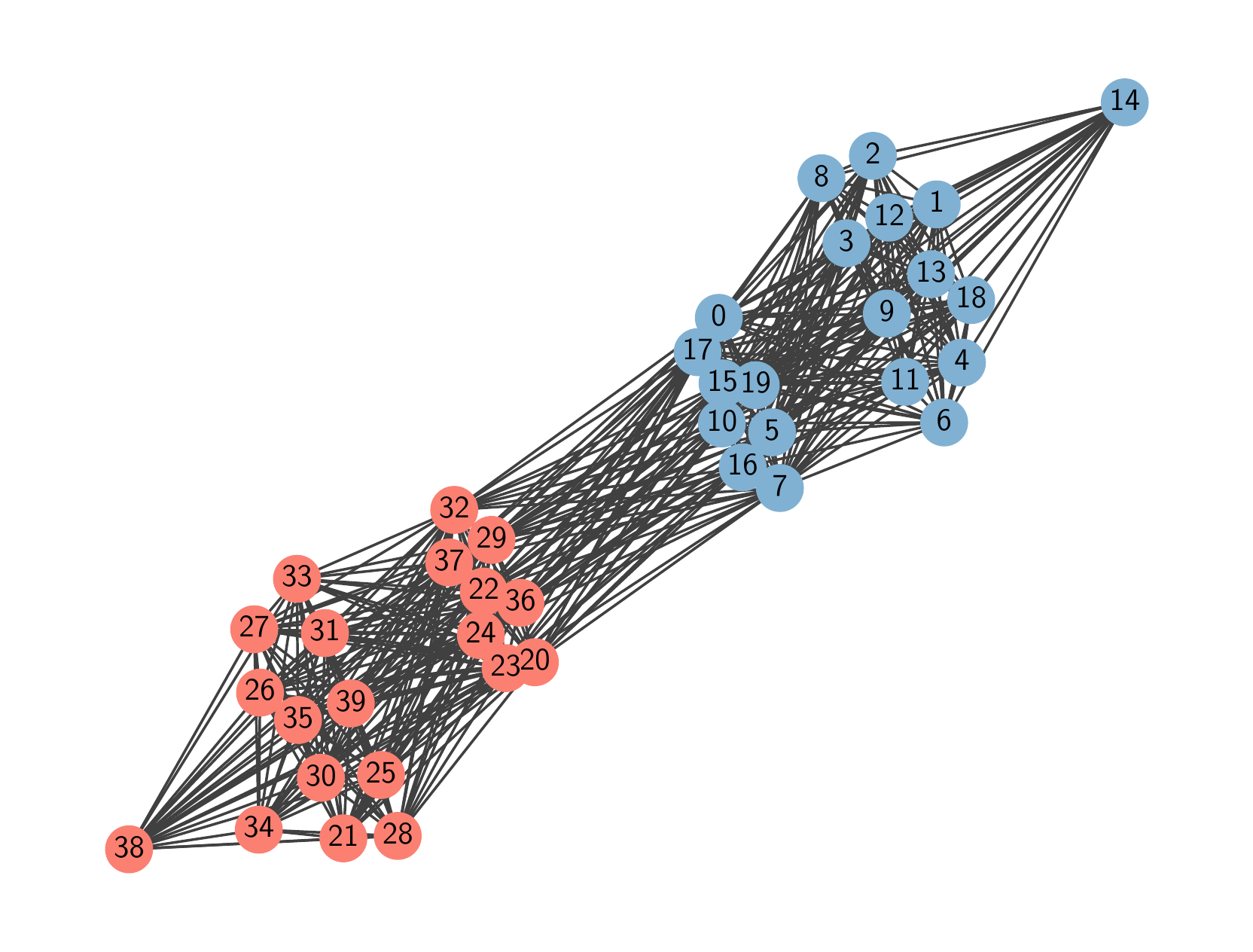} &
        			\includegraphics[width=0.4\linewidth]{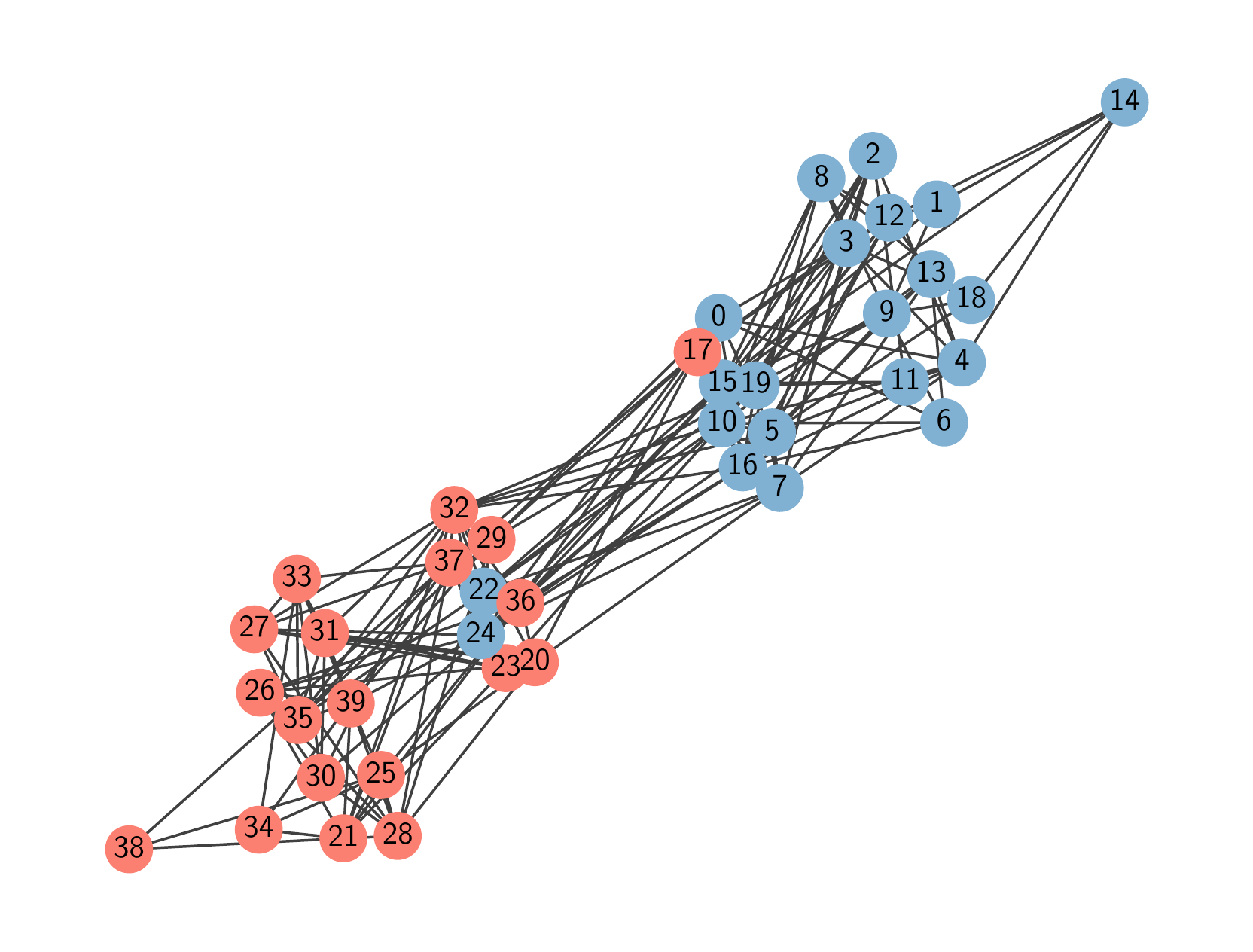}  \vspace{-4mm}\\
        			{\bf (a)} & {\bf (b)} \\
        			\includegraphics[width=0.4\linewidth]{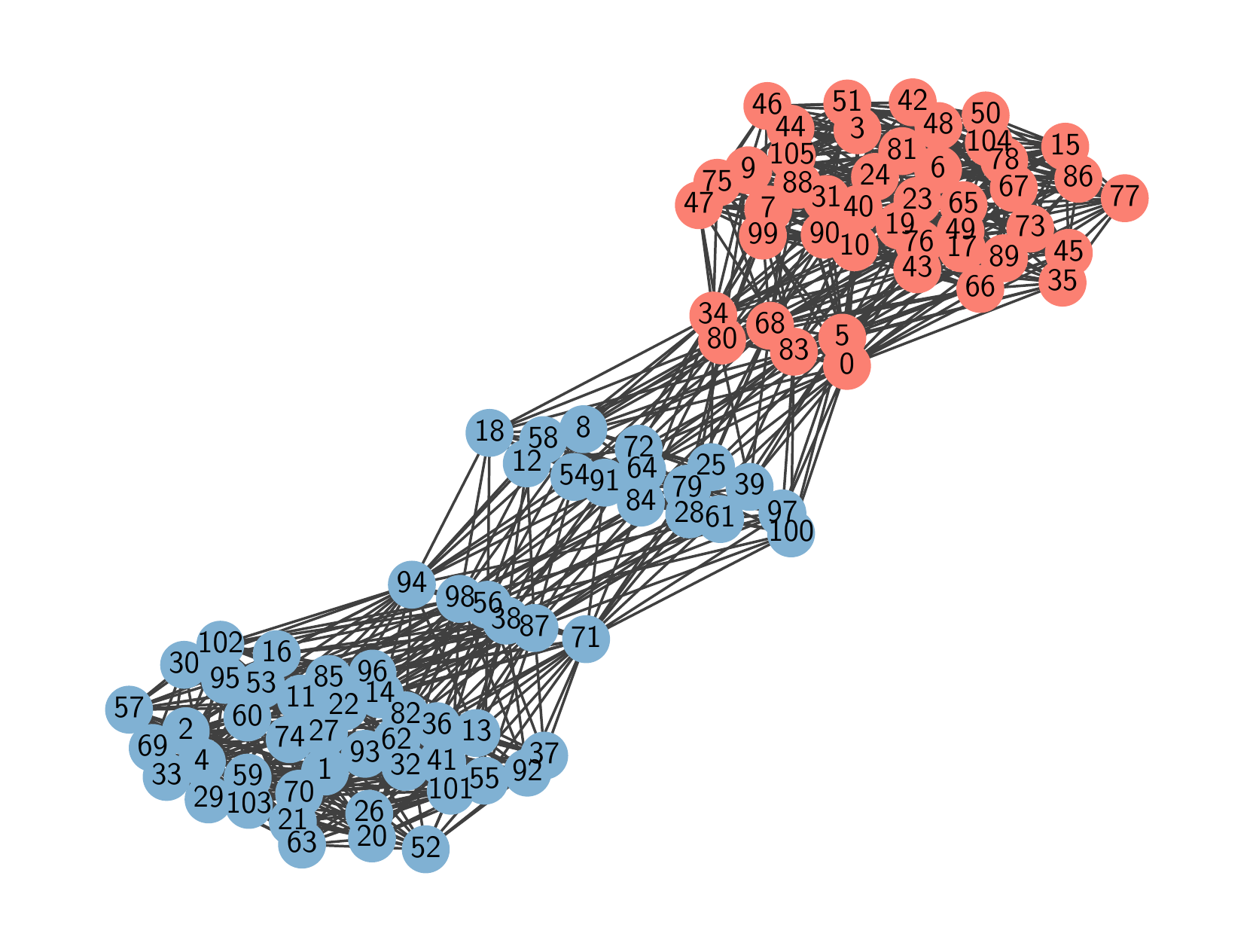} &
        			\includegraphics[width=0.4\linewidth]{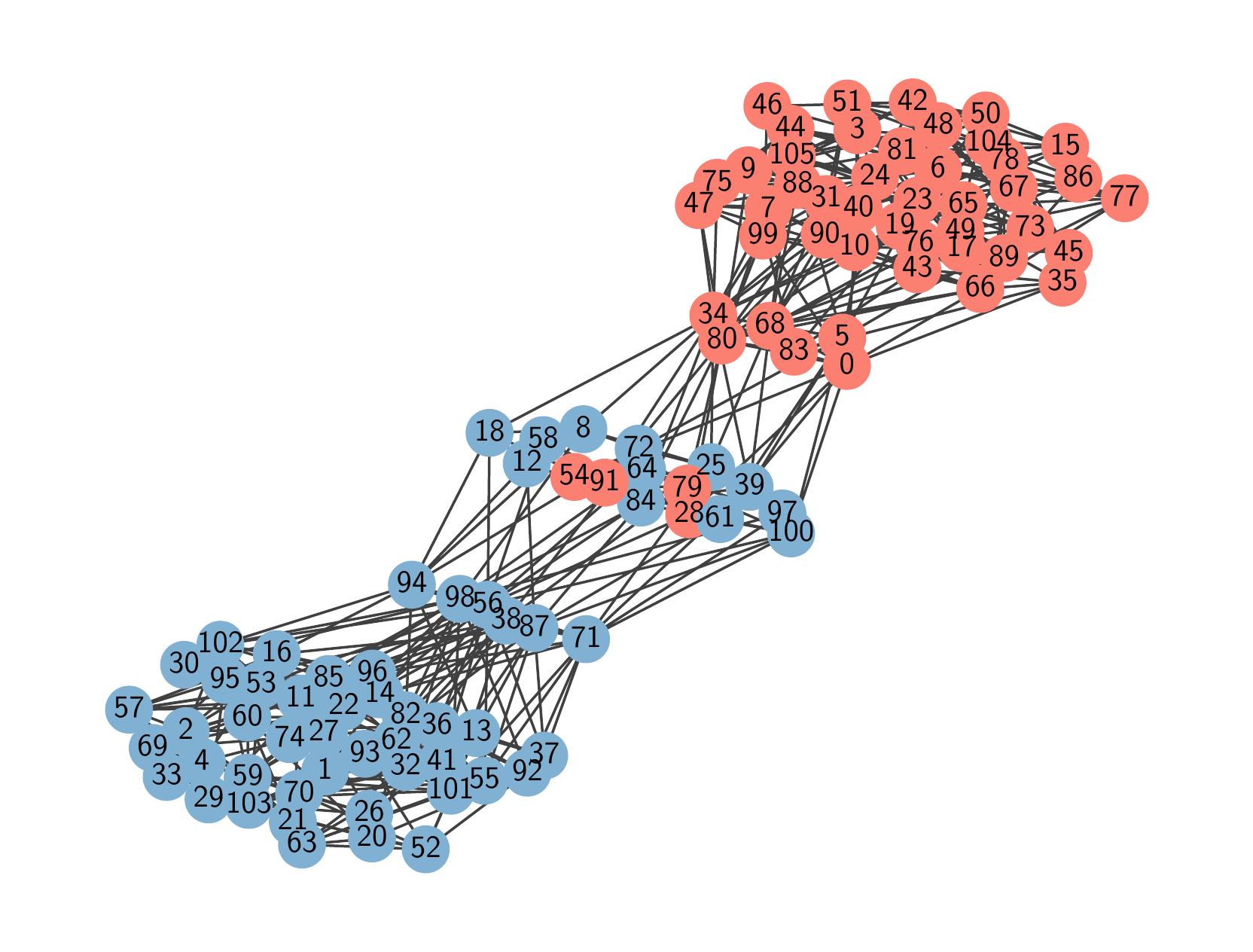} \vspace{-4mm}\\
        			{\bf (c)} & {\bf (d)}
        		\end{tabular}
        		\vspace{-2mm}
        		\caption{{\bf (a)-(b)}: Spectral clustering of graphs before {\bf (a)} and after {\bf (b)}  sparsification. {\bf (c)-(d)}: Spectral clustering of simplicial complexes into two clusters, before {\bf (c)} and after {\bf (d)} sparsification. We observe that the clusters are very similar. See Section \ref{subsec:spectral-clustering} for details. }
        		\label{fig:clustering-graphs-sc}
        	\end{center}
        \end{figure}
    
    \para{Spectral clustering algorithm for simplicial complexes.}
        We seek to extend the Ng-Jordan-Weiss algorithm~\cite{NgJordanWeiss2001} to simplicial complexes, which, as far as we are aware, has not yet been studied.
        We seek the simplest generalization by replacing the vertex-vertex affinity matrix with an edge-edge affinity matrix $A$, where two edges are considered to be adjacent if they are faces of the same triangle.
        This definition is a straightforward extension of the adjacency among vertices in graphs, however it does not account for the orientation of edges or triangles.
        
        Formally, let $n_1$ be the number of edges.
        Let $W_2$ be the diagonal weight matrix for triangles.
        We define the edge-edge \emph{affinity matrix} $A\in \Rspace^{n_1 \times n_1}$, where
        
        \[ A_{i,j} = \begin{cases}
        w_f & \textrm{if edges $e_i$ and $e_j$ are adjacent to triangle $f$ with weight $w_f$ in $W_2$}\\ 
        0 & \textrm{otherwise} 
        \end{cases} . \]
        We define $\Delta \in \Rspace^{n_1 \times n_1}$ to be the diagonal matrix with element $\Delta_{i,i}$ being the sum of $A$'s $i$-th row.
        With $A$ and $\Delta$ defined this way, we can apply the Ng-Jordan-Weiss algorithm to cluster the edges of the simplicial complex $K$.
        
        This is  equivalent to applying spectral clustering to the dual graph of $K$.
        A \emph{dual graph} $G$ of a given simplicial complex $K$ is created as follows: each edge in $K$ becomes a vertex in the dual graph $G$, and there is an edge between two vertices in $G$ if their corresponding edges in $K$ share the same triangle.
        We then apply spectral clustering to the dual graph $G$ as usual and obtain the resulting clustering of vertices in $G$ (which correspond to the clustering of edges in $K$).
        To better illustrate our edge clustering results, we visualize the resulting clusters based upon the dual graph.
        The results are plotted in Figure~\ref{fig:clustering-graphs-sc} (c)-(d) for two clusters and Figure~\ref{fig:clustering-sc-3} for three clusters.
        Applying the spectral algorithm with these new definitions of $A$ and $\Delta$ results in clusters that agree reasonably well before and after sparsification.
        
        \begin{figure}[!ht]
        	\vspace{-2mm}
        	\begin{center}
        		\begin{tabular}{cc}
        			\includegraphics[width=0.4\linewidth]{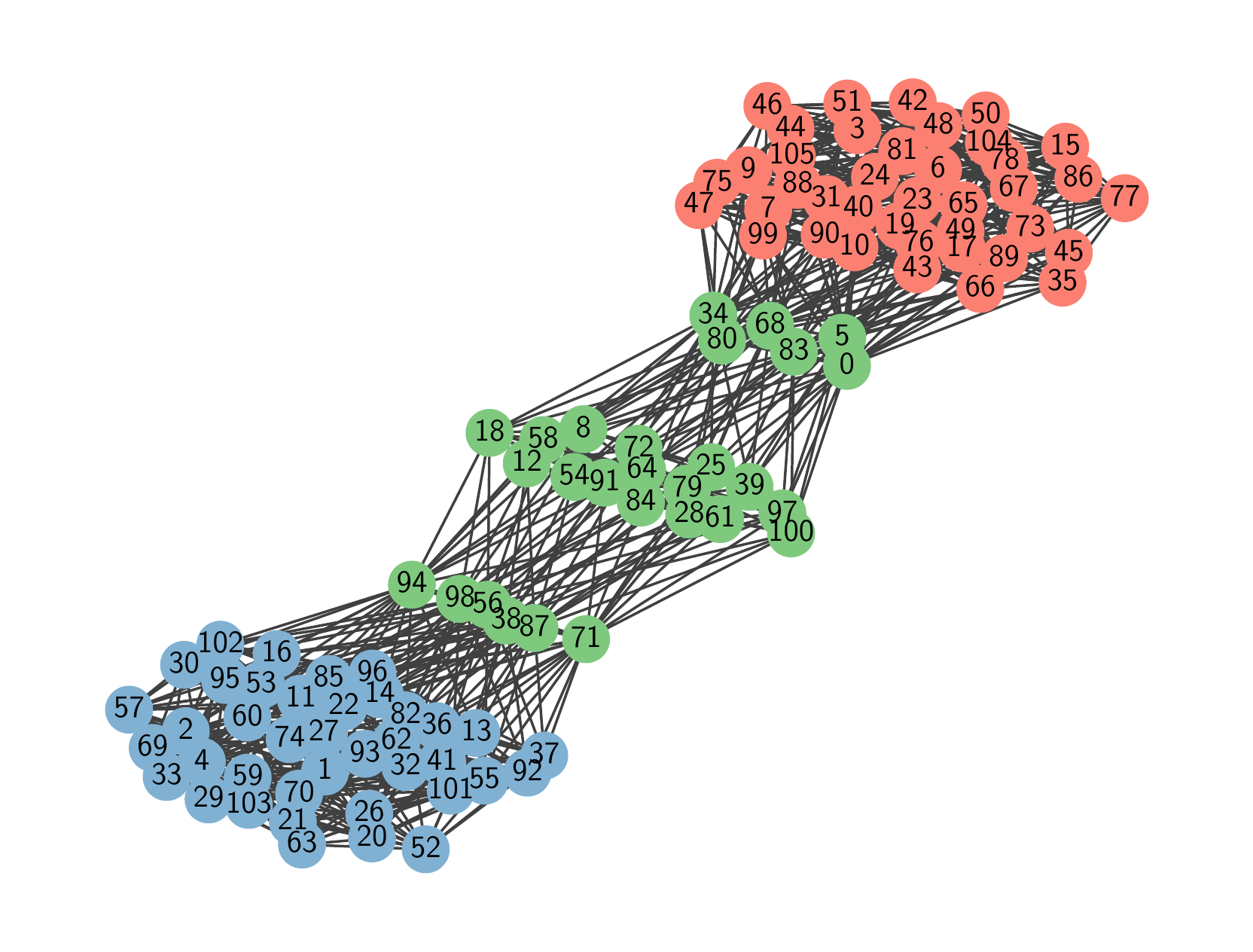} &
        			\includegraphics[width=0.4\linewidth]{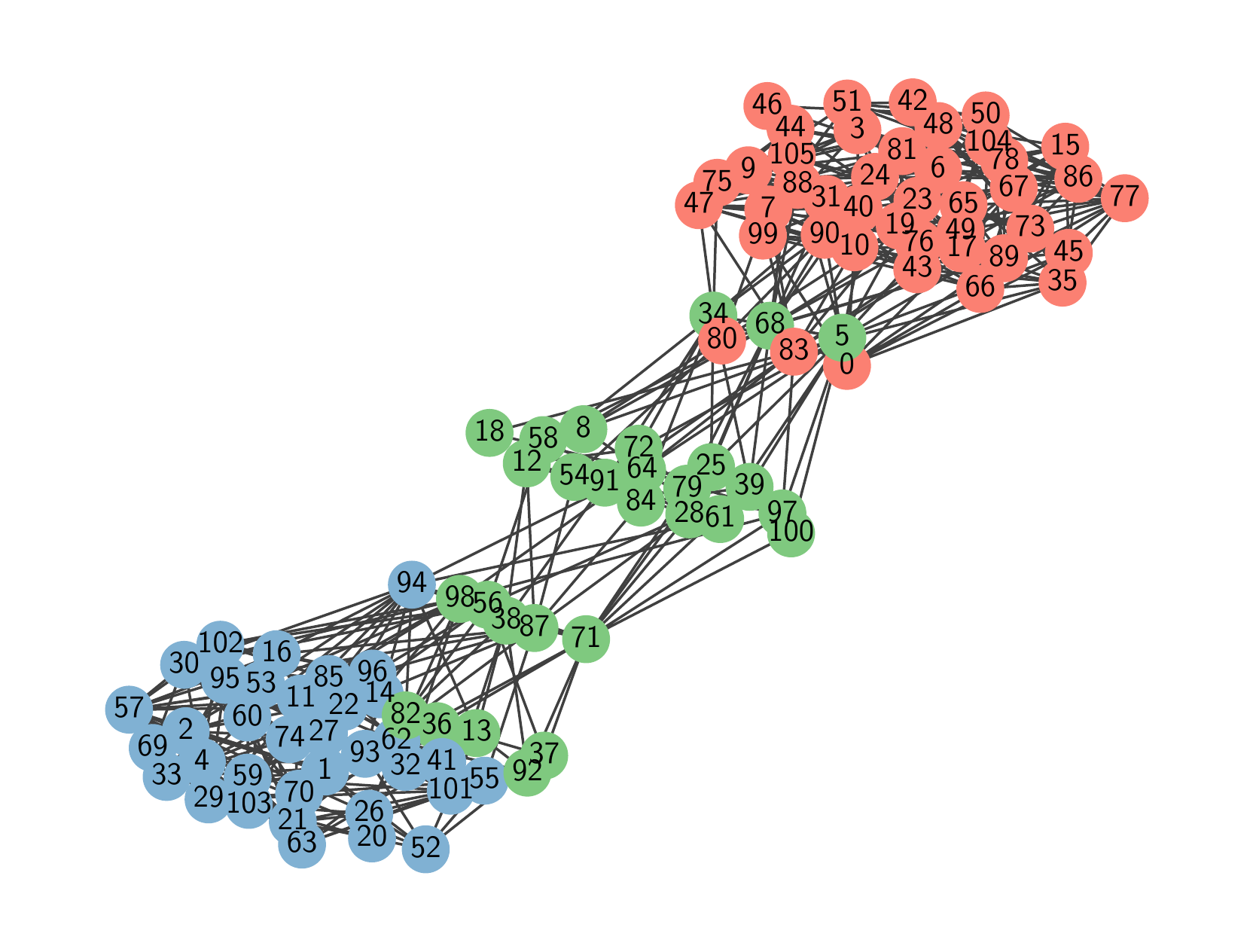} \vspace{-4mm}\\
        			{\bf (c)} & {\bf (d)}
        		\end{tabular}
        		\vspace{-2mm}
        		\caption{Spectral clustering of simplicial complexes into three clusters, before {\bf (a)} and after {\bf (b)} sparsification. See Section~\ref{subsec:spectral-clustering} for details.}
        		\label{fig:clustering-sc-3}
        	\end{center}
        \end{figure}
        
        The affinity matrix, $A$, does not take into consideration the orientation of the edges, so the above clustering algorithm does not directly rely on the up Laplacian.
        One can verify that the dimension $1$ up Laplacian can be written as 
        $\Lcal_{K,1} = \Delta/2 - A^*$, 
        where $\Delta$ is the diagonal matrix defined previously
        and the \emph{oriented edge-edge affinity matrix}, 
        $A^* \in \Rspace^{{n_1} \times {n_1}}$, is given by
        
        \[ A^*_{i, j} = 
        \begin{cases}
        -w_f & \textrm{edges $e_{i}$ and $e_{j}$ are both faces of the same triangle $f$ and both agree or } \\
        & \textrm{disagree with the orientation of their shared triangle} \\
        w_f & \textrm{if either $e_{i}$ or $e_{j}$ (but not both) agree with the orientation of $f$} \\
        0 & \textrm{if $e_{i}$ and $e_{j}$ are not adjacent}
        \end{cases}. \]
        It follows that $A = |A^*|$ where the absolute value operation is applied element-wise.
        The relation between $\Delta - A$ and the up Laplacian, $\Lcal_{K,1} $, we used for sparsification, remains unclear.


\subsection{Label propagation}
\label{subsec:label-propagation}
    
    A good example of spectral methods in learning arises from extending label propagation algorithms on graphs to simplicial complexes, in particular, the work by Mukherjee and Steenbergen~\cite{MukherjeeSteenbergen2016}.
    Specifically, they adapt the label propagation algorithm to higher dimensional walks on oriented edges, and give visual examples of applying label propagation with the $1$-dimensional up Laplacian $\Lcal_1^{up}$, down Laplacian $\Lcal_1^{down}$, and Laplacian $\Lcal_1.$
    We envision label propagation to be generalized to random walks on even higher-dimensional simplices, such as triangles.
    A direct application of our work is to sparsify the top-dimensional simplices (e.g. triangles in a $2$-dimensional simplicial complex) and examine how label propagation behaves on these top-dimensional simplices of the sparsified representation.
    
    Similar to the setting of Section~\ref{subsec:spectral-clustering}, we apply and generalize a simple version of label propagation algorithms~\cite{ZhuGhahramani2002} to the setting of both graphs and simplicial complexes.
    In particular, as illustrated in Figure~\ref{fig:sc-lp}, we show via the dual graph representation that the results obtained from sparsified simplicial complexes are similar to those of the original simplicial complex.
    We now describe the algorithmic details.
    
    \begin{figure}[t!]
    	\begin{center}
    		\begin{tabular}{cc}
    			\includegraphics[width=0.49\linewidth]{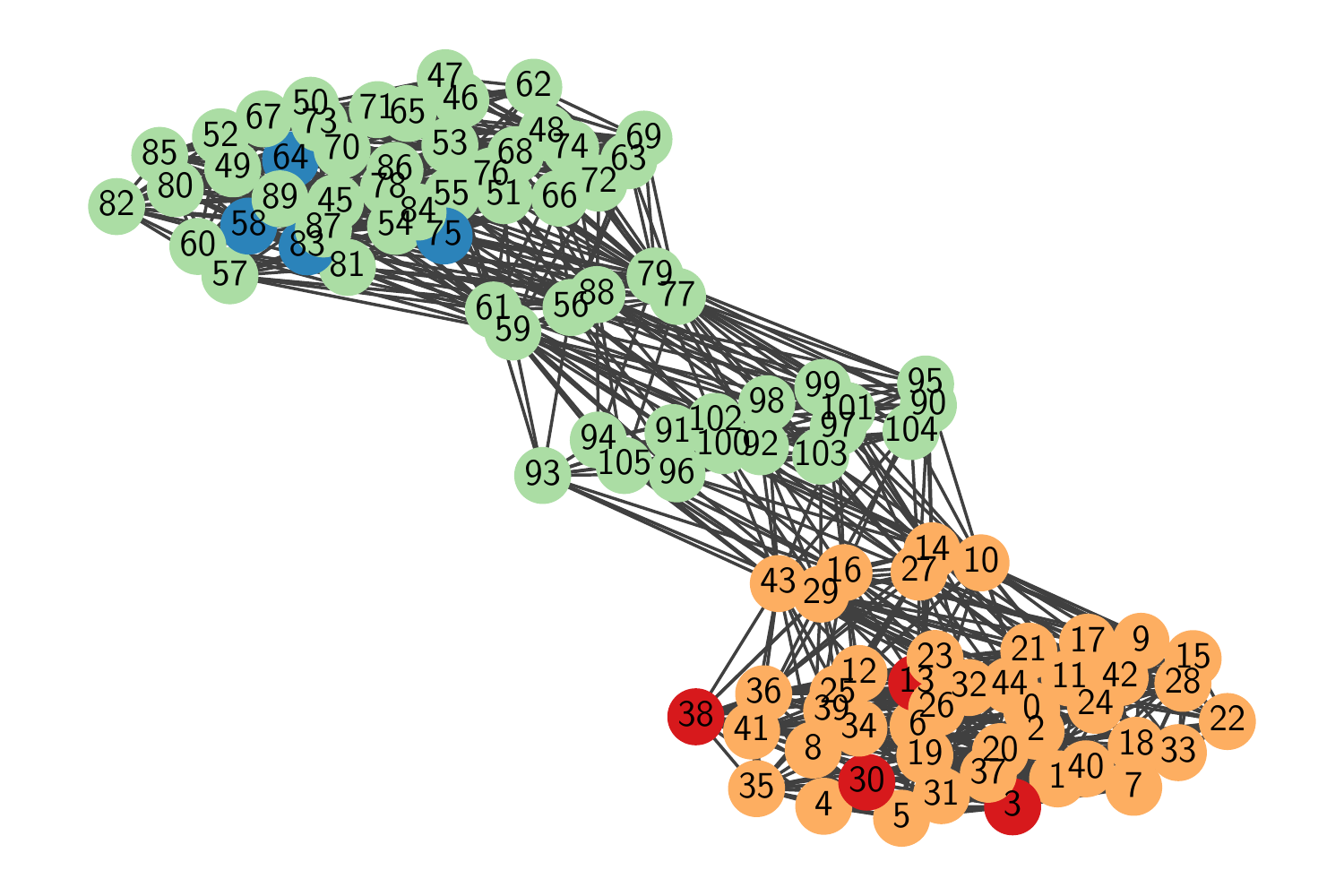} &
    			\includegraphics[width=0.49\linewidth]{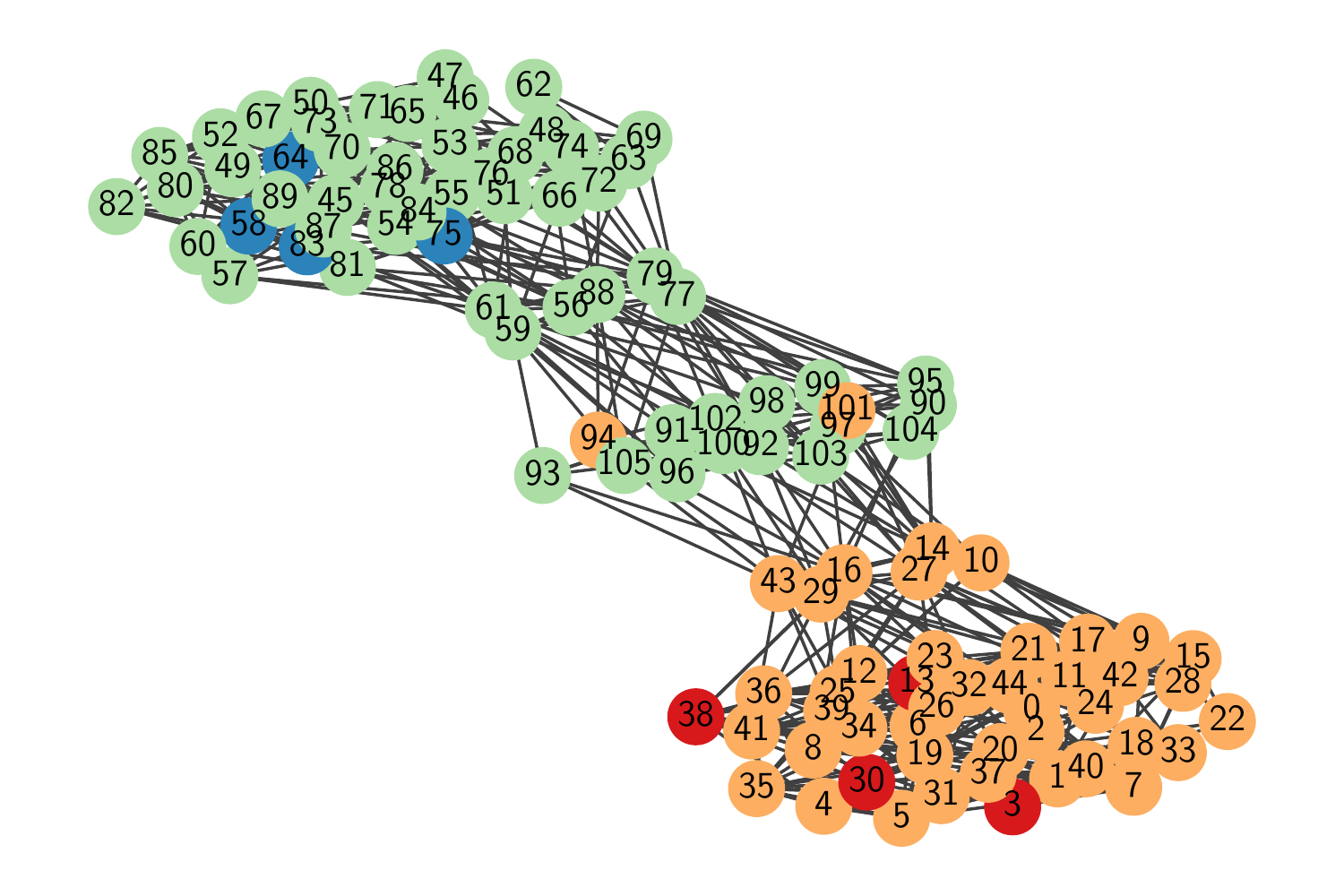} \vspace{-4mm}\\
    			{\bf (a)} & {\bf (b)}
    		\end{tabular}
    		\vspace{-2mm}
    		\caption{The results of label propagation on simplicial complexes before {\bf (a)} and after {\bf (b)} sparsification. The red and blue colored vertices correspond to fixed edge labels while the green and orange colored vertices  correspond to propagated edge labels. Blue and green share the same label while red and orange share the same label. 
    			See Section~\ref{subsec:label-propagation} for details.}
    		\label{fig:sc-lp}
    	\end{center}
    \end{figure}
    
    \para{Label propagation on graphs.}
        We implement a simple version of the iterative label propagation algorithm~\cite{ZhuGhahramani2002} based on the notion of stochastic matrix (i.e.~random walk matrix) $P = \Delta^{-1}A$, where $A$ is the affinity matrix and $\Delta$ is the diagonal matrix with diagonal elements $\Delta_{ii} = \sum_{j}a_{ij}$ (as defined in Section~\ref{subsec:spectral-clustering}).
        
        The matrix $P$ represents the probability of label transition.
        Given $P$ and an initial label vector $\mathbf{y}$, we iteratively multiply the label vector $\mathbf{y}$ by $P$. If the graph is \emph{label-connected} (i.e.~we can always reach a labeled vertex from any unlabeled one), then $P^t$ converges to a stationary distribution, that is,  $P^t \mathbf{x} = \mathbf{x}$ for a large enough $t$.
        
        Suppose there are two label classes $\{+1, -1\}$.
        Without loss of generality, assume that first $l$ of the $n$ vertices are assigned labels initially, represented as a length-$l$ vector $\mathbf{y}_l$. Given a graph $G(V, E)$ and labels $\mathbf{y}_l$, the algorithm is given as:
        \begin{enumerate}
        	\item Compute $A$, $\Delta$, and $P = \Delta^{-1}A$.
        	\item Initialize $\mathbf{y}^{(0)} = (\mathbf{y}_l, \mathbf{0})$, $t = 0$.
        	\item Repeat until convergence:
        
        	$$\mathbf{y}^{(t+1)} = P \mathbf{y}^{(t)},$$
        	$$\mathbf{y}_l^{(t+1)} = \mathbf{y}_l^{(t)}.$$
        	\item Return $\sgn(\mathbf{y^{(t)}}).$
        \end{enumerate}
        Consider $P$ to be divided into blocks as follows:
        
        $$P = \begin{pmatrix}
        P_{ll} & P_{lu} \\
        P_{ul} & P_{uu}, 
        \end{pmatrix}$$
        where $l$ and $u$ index the labeled and unlabeled vertices with the number of vertices $n_0 = l + u$. 
        Let $\mathbf{y} = (\mathbf{y}_l, \mathbf{y}_u)$ be the labels at convergence, then $\mathbf{y}_u$ is given by :
        
        $$\mathbf{y}_u = (I - P_{uu})^{-1}P_{ul}\mathbf{y_l}$$
        As long as our graph is connected, it is also label-connected and $(I - P_{uu})$ is non-singular. So we can directly compute the labels at convergence without going through the iterative process described above.
        
        \begin{figure}[ht!]
        	\begin{center}
        		\begin{tabular}{cc}
        			\includegraphics[width=0.49\linewidth]{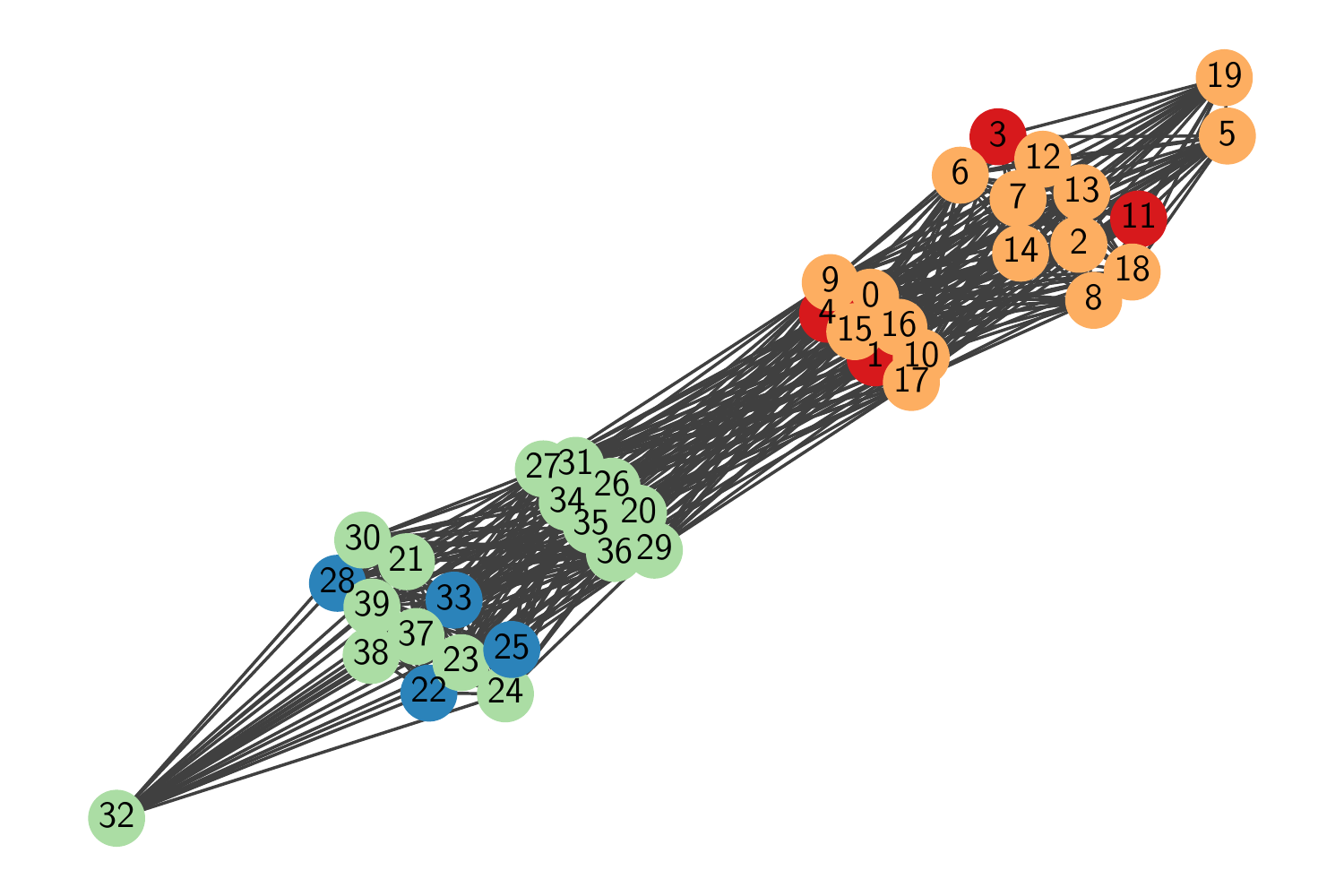} &
        			\includegraphics[width=0.49\linewidth]{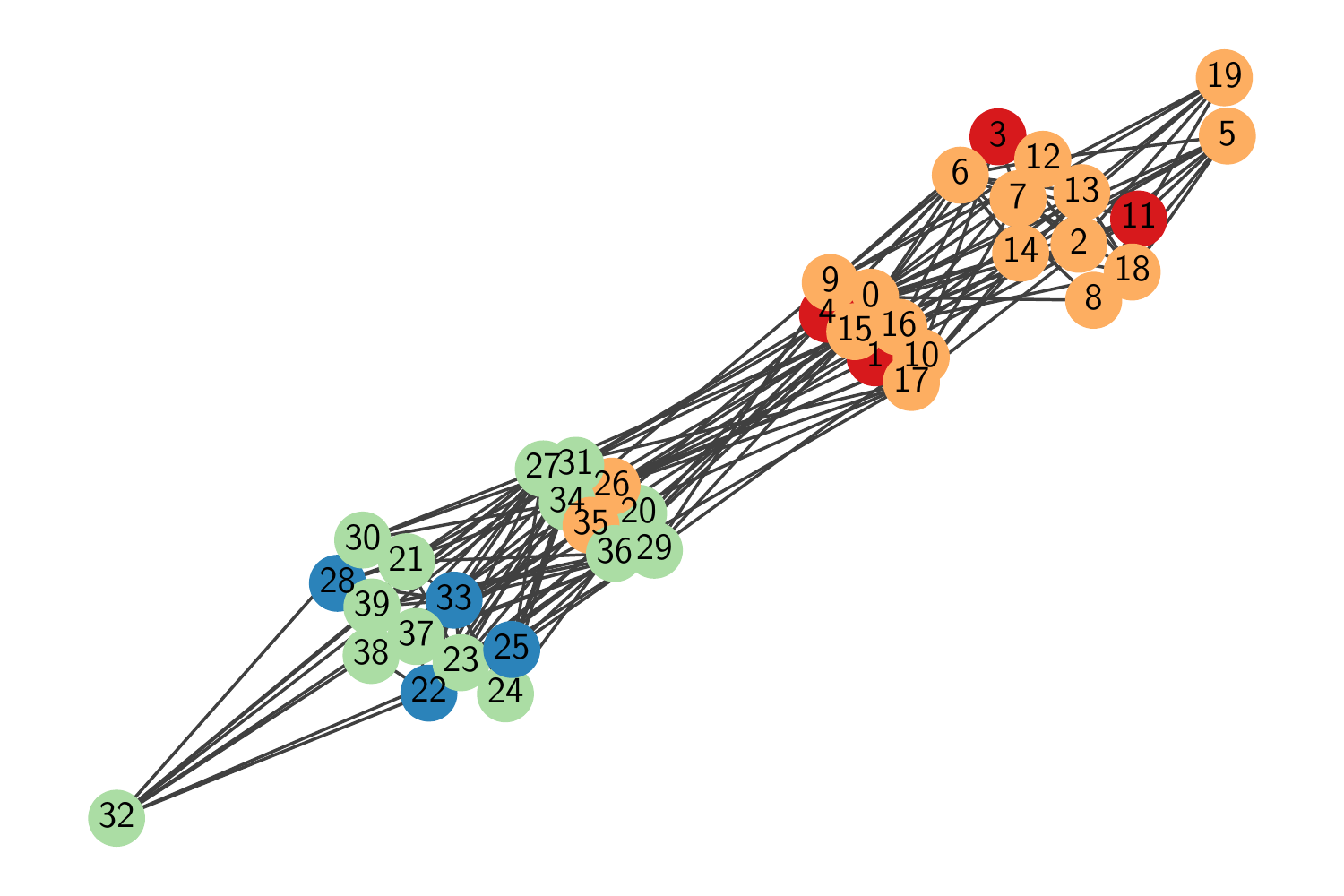} \vspace{-4mm}\\
        			{\bf (a)} & {\bf (b)}
        		\end{tabular}
        		\vspace{-2mm}
        		\caption{The results of label propagation on the dumbbell graph before {\bf (a)} and after {\bf (b)} sparsification. The red and blue color represent the initial opposite vertex labels while the green and orange color correspond to the final propagated vertex labels. Blue and green share the same label while red and orange share the same label.}
        		\label{fig:graph-lp}
        	\end{center}
        \end{figure}
        As illustrated in Figure~\ref{fig:graph-lp}, we apply label propagation algorithm to the dumbbell graph dataset to demonstrate that preserving the structure of graph Laplacian via sparsification also preserves the results of label propagation on graphs.
    
    \para{Label propagation on simplicial complexes.}
        To apply label propagation to our dumpbell complex example, we could extend the label propagation algorithm of~\cite{ZhuGhahramani2002} to simplicial complexes, again, by replacing the vertex-vertex affinity matrix with edge-edge affinity matrix $A$. 
        As a consequence, the new diagonal matrix $\Delta$ and the stochastic matrix $P$ capture relations among edges instead of vertices. 
        Without considering the orientation of edges or triangles, the algorithm can be considered as applying label propagation to the dual graph of the simplicial complex. 
        
        In addition to the example showing in Figure~\ref{fig:sc-lp}, we give a few more instances of the results of label propagation on the dumbbell complex in Figure~\ref{fig:sc-lp-more} with different initial labels. 
        
        \begin{figure}[t!]
        	\begin{center}
        		\begin{tabular}{cc}
        			\includegraphics[width=0.49\linewidth]{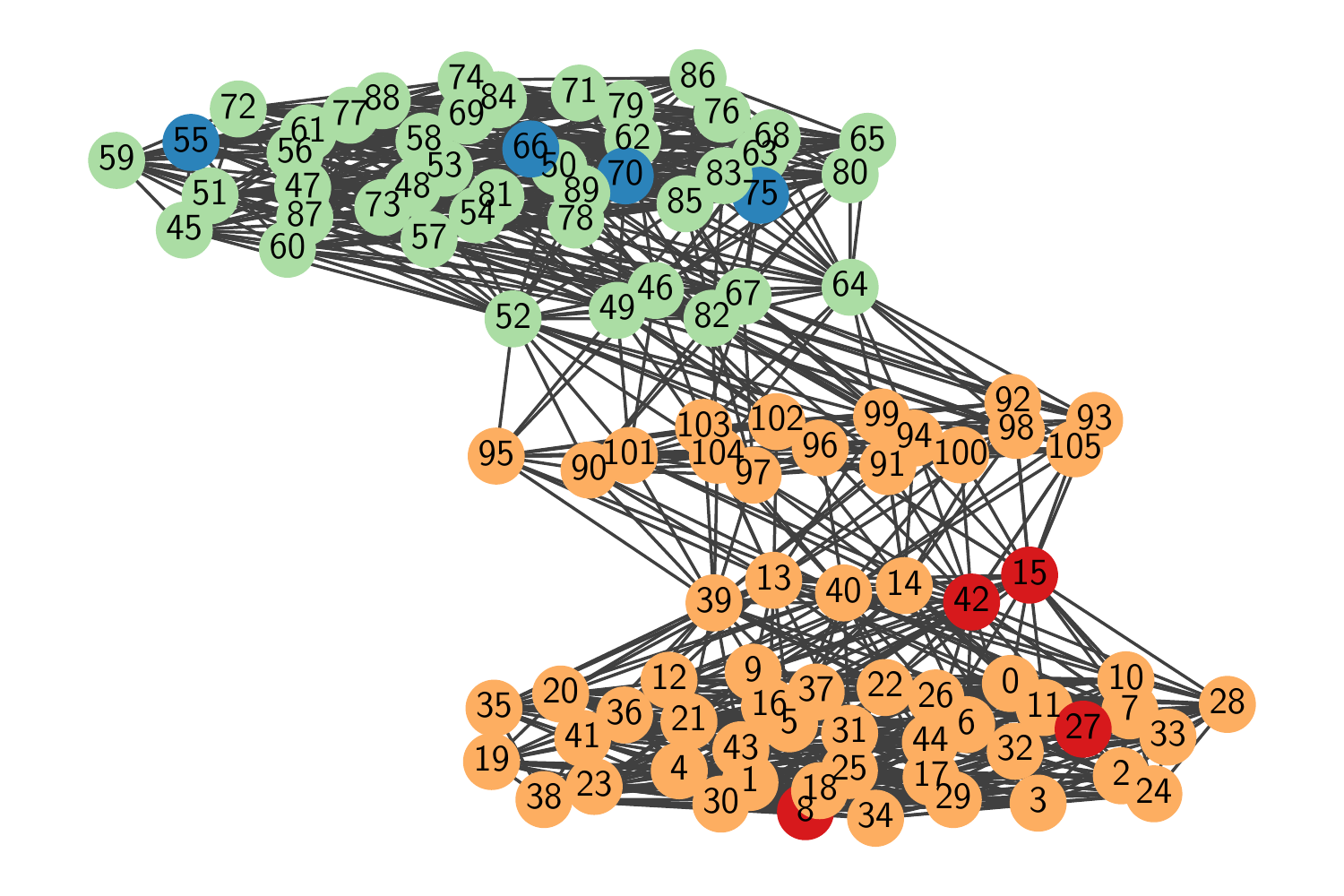} &
        			\includegraphics[width=0.49\linewidth]{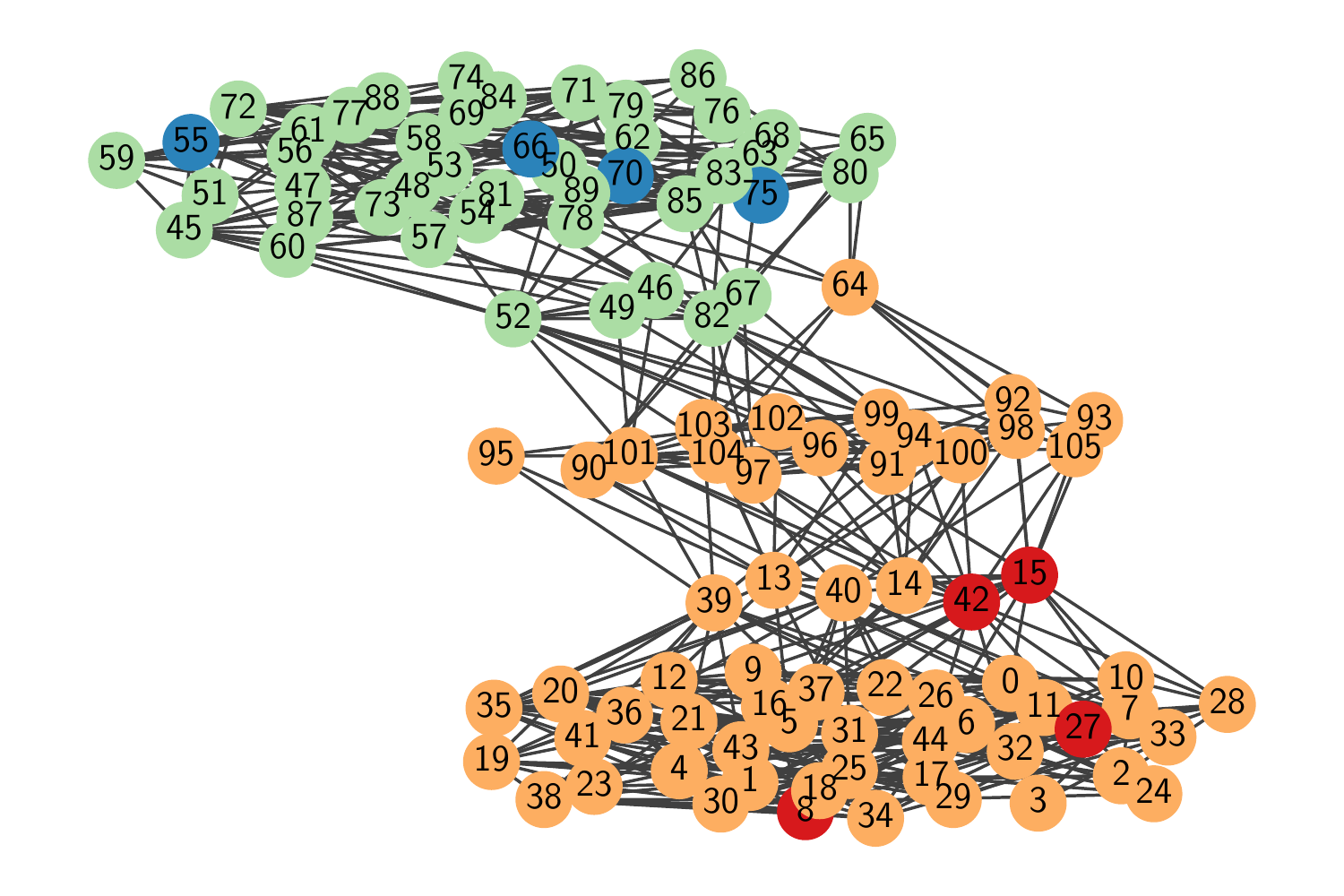} 
        			\vspace{-4mm}\\
        			{\bf (a)} & {\bf (b)}\\
        			\includegraphics[width=0.49\linewidth]{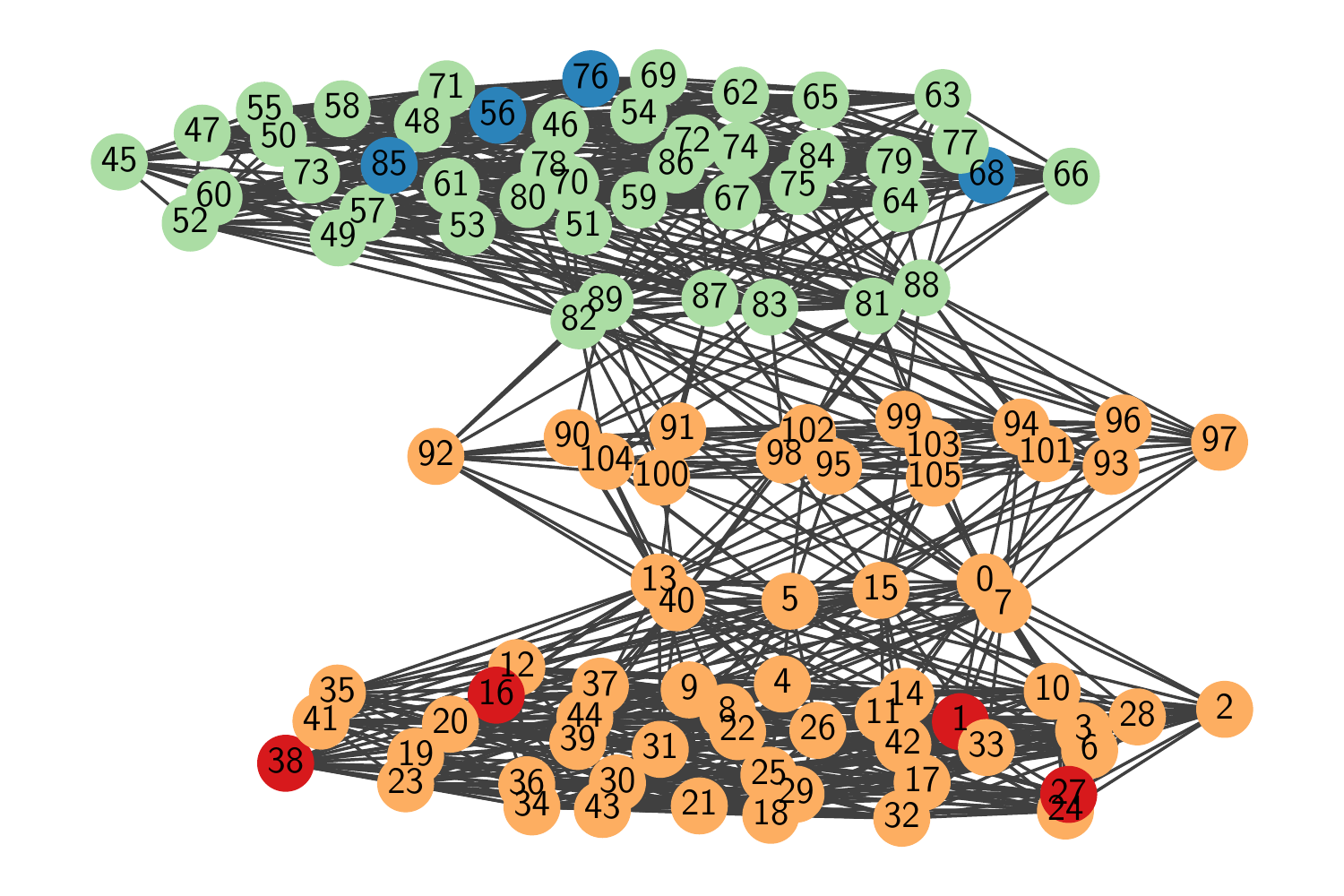} &
        			\includegraphics[width=0.49\linewidth]{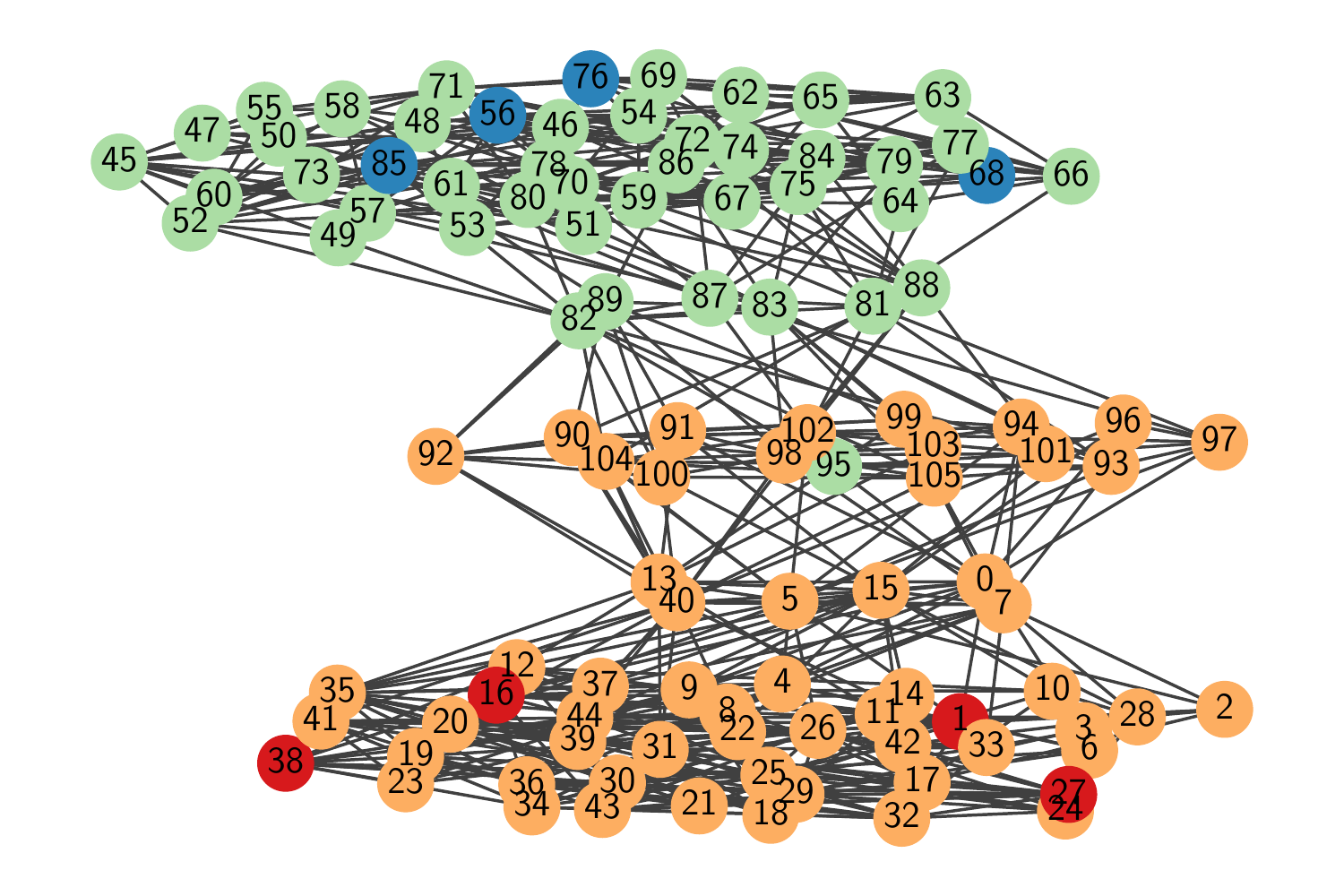} 
        			\vspace{-4mm}\\
        			{\bf (c)} & {\bf (d)}\\
        		\end{tabular}
        		\vspace{-2mm}
        		\caption{More instances of label propagation on the dumbbell complex before {\bf (a), (c)} and after {\bf (b), (d)} sparsification.}
        		\label{fig:sc-lp-more}
        	\end{center}
        \end{figure}

	
\section{Discussion}
\label{sec:discussion}

    We presented an algorithm for the simplification of simplicial complexes that preserves spectral properties of the up Laplacian.
    Our work is strongly motivated by the study of an emerging class of learning algorithms based on simplicial complexes and, in particular, those spectral algorithms that operate with higher-order Laplacians.
    We would like to understand the benefits and incurred error when such learning algorithms are applied to sketches of the data.
    Several on-going and future directions are described below.
    
    \para{Physical meaning of generalized effective resistance.}
        We believe the generalization of effective resistance to simplicial complexes, introduced in Section \ref{sec:algorithm}, may find other applications in analyzing simplicial complexes.
        Though the generalization is algebraically straightforward, there are many natural and interesting questions about its interpretation and properties.
        For example, does it have an interpretation in terms of a random process, such as an effective commute time as in the case of a graph (see, e.g., ~\cite{GhoshBoydSaberi2008})?
        Is it related to minimum spanning objects in the simplicial complex?  Does it play a further role in spectral clustering of simplicial complexes?
    
    \para{Multilevel and Hodge sparsification.} 
        We are also interested in performing multilevel sparsification of simplicial complexes; for example, we would like to sparsify triangles and edges simultaneously while preserving spectral properties of the dimension-$0$ and dimension-$1$ up Laplacians.
        This is  challenging if we would like to simultaneously maintain structures of simplicial complexes; it may be possible if we could relax our structural constraints to work with hyper-graphs instead.
        In addition, multilevel sparsification is also related to preserving the spectral properties of the (Hodge) Laplacian.
        Finally, we are also interested in deriving formal connections between homological sparsification and spectral sparsification of simplicial complexes.
	
	
\section*{Acknowledgements}
	This work was partially supported by NSF DMS-1461138, NSF IIS-1513616, and the University of Utah Seed Grant 10041533.
	We would like to thank Todd H. Reeb for contributing to early discussions.


\newpage

	
\newpage
\appendix
\section{Complexity}
\label{sec:app-complexity}

    Suppose we are given a weighted, oriented simplicial complex $K$ and a fixed dimension $i$ where $1 \leq i \leq \dim(K)$. We will denote the number of $i$-simplices of $K$ as $n_i$.


\subsection{Na\"{i}ve implementation}

    \para{Sparsification.}
        To sparsify $K$ at dimension $i$, our algorithm needs to compute the incidence matrix $D_{i-1}$, the up Laplacian $\mathcal{L}_{K, i-1}$, the Moore-Penrose inverse of the up Laplacian $(\mathcal{L}_{K, i-1})^{+}$ and the generalized effective resistance matrix $R_i$.
        
        Computing the incidence matrix, $D_{i-1}$, requires a constant number of operations per $i$-simplex, $O(n_{i})$.
        Computing $\mathcal{L}_{K, i-1}$, requires two matrix-matrix multiplications, one of which involves the diagonal matrix $W_{i}$, $O(n_{i}n_{i-1})$.
        The up Laplacian computed is an $n_{i-1} \times n_{i-1}$ symmetric positive semidefinite matrix.
        In the most na\"{i}ve implementation, we compute the Moore-Penrose pseudo-inverse by using a QR decomposition routine which requires $O(n_{i-1}^3)$ number of operations.
        Once we have the inverse, computing $R_i$ again takes $O(n_{i}n_{i-1})$. Since, $n_{i} \leq n_{i-1}^2$ for any simplicial complex, the overall complexity scales as that of computing the inverse, that is,  $O(n_{i-1}^3)$.
    
    \para{Spectral Clustering.}
        In spectral clustering, our objective is to cluster $(i-1)$-simplices of $K$ into $k$ clusters.
        To do this, our algorithm first computes the eigenvectors corresponding to the $k$ largest eigenvalues of $\mathcal{L}_{K, i-1}$ and then applies $k$-means clustering to the point set of size $n_{i-1}$ in $\mathbb{R}^k$ formed by these $k$ eigenvectors.
        A na\"{i}ve algorithm to compute $k$ eigenvectors of an $n_{i-1} \times n_{i-1}$ matrix requires $O(n_{i-1}^2k)$ operations.
        The $k$-means algorithm (Lloyd's algorithm) to cluster $n_{i-1}$ points in $\mathbb{R}^k$ into $k$ clusters runs in $O(n_{i-1}k^2 j)$, where $j$ is the number of iterations required for convergence. We may assume, in general, that $k << n_{i-1}$.
        Therefore, overall complexity of spectral clustering is $O(n_{i-1}^2k)$.
    
    \para{Label Propagation.}
        In the label propagation problem, we are given discrete labels for a small subset of $(i-1)$-simplices of $K$ and the objective is to learn the labels on remaining unlabeled $(i-1)$-simplices.
        Our label propagation algorithm requires computing the transition probability matrix $P$ by normalizing the adjacency matrix of $(i-1)$-simplices.
        Then, it computes the inverse of the sub-matrix $P$ corresponding to the set of unlabeled edges.
        Assuming the number of labeled edges is small, computing inverse requires $O(n_{i-1}^3)$ operations.
        After that, the algorithm only requires two matrix-vector multiplications $O(n_{i-1}^2)$.
        So the overall complexity of label propagation is $O(n_{i-1}^3)$.


\subsection{Sparse matrix implementation}

    \para{Sparsification.}
        Note that unless we are dealing with a complete simplicial complex, the up Laplacian $\mathcal{L}_{K, i-1}$ is fairly sparse, and as such, we can use algorithms specifically designed to handle sparse matrices.
        In our implementations, we used SciPy's sparse linear algebra module.
        This can significantly reduce the number of operations required to perform all the matrix-matrix multiplications.
        However, computing the Moore-Penrose pseudo-inverse $(\mathcal{L}_{K, i-1})^+$ still requires $O(n_{i-1}^3)$ operations.
        The SciPy implementation for pseudo-inverse uses the QR decomposition.
    
    \para{Spectral Clustering.}
        The sparse eigenvalue solver of SciPy uses ARPACK's Implicitly Restarted Arnoldi Method (IRAM).
        The rate of execution (in flops) for an IRAM iteration is asymptotic to the rate of execution of matrix-vector multiplication routine of BLAS.
        That is, computing $k$ eigenvectors requires $O(\nnz \cdot k \cdot s)$  where $\nnz$ is the number of non-zero entries in $\mathcal{L}_{K, i-1}$ and $s$ is the number of iterations required for convergence.
        Once the eigenvectors are computed, the $k$-means (Lloyd's) algorithm runs in $O(n_{i-1} \cdot k^2 \cdot t)$ where $t$ is the number of iterations required for $k$-means algorithm to converge.

    \para{Label Propagation.}
        Our label propagation algorithm requires solving the following linear system 
        $(I - P_{uu})y_u = P_{ul}y_l$,  
        where $P$ is the normalized adjacency matrix (transition probability matrix) of $(i-1)$-simplices of $K$, $y_l$ is the vector of known labels, and $P_{uu}$ is the sub-matrix of $P$ corresponding to unlabeled $(i-1)$-simplices.
        As long as the simplices are label-connected (there is a sequence of $i$-simplices connecting every unlabeled $(i-1)$-simplex to a labeled $(i-1)$-simplex), $(I - P_{uu})$ is symmetric positive definite.
        Using the sparse implementation of conjugate gradient, the system can be solved in $O(\nnz \cdot n_{i-1})$ where $\nnz$ is the number of non-zero entries in $P$ which is the same as the number of non-zero entries in the adjacency matrix of $(i-1)$-simplices of $K$.
	
	
\end{document}